\documentclass[twocolumn, 10pt]{IEEEtran}
\usepackage{amsmath,pdfpages}
\usepackage{bm,graphicx}
\usepackage{algorithm}
\usepackage{color}

\newcommand {\cD}{{\mathcal{D}}}

\newcommand {\cI}{{\mathcal{I}}}

\newcommand {\cP}{{\mathcal{P}}}

\newcommand {\cX}{{\mathcal{X}}}

\newcommand {\ba} {{\bf a}}
\newcommand {\bff} {{\bf f}}
\newcommand {\bg} {{\bf g}}
\newcommand {\bs} {{\bf s}}

\newcommand {\bx} {{\bf x}}

\newcommand {\bw} {{\bf w}}

\newcommand{\avg}{{\rm avg}}
\newcommand{\opt}{{\rm opt}}

\newcommand {\N} {{\rm I\kern-1.5pt N}}
\newcommand {\R} {{\rm I\kern-2.5pt R}}

\newtheorem{lemma}{Lemma}

\newtheorem{coro}{Corollary}
\newtheorem{theorem}{Theorem}
\newtheorem{defn}{Definition}
\newtheorem{assm}{Assumption}

\newcommand{\beqa}{\begin{eqnarray}}
\newcommand{\eeqa}{\end{eqnarray}}
\newcommand{\beqan}{\begin{eqnarray*}}
\newcommand{\eeqan}{\end{eqnarray*}}
\newcommand{\beq}{\begin{equation}}
\newcommand{\eeq}{\end{equation}}

\newcommand{\bfl}{\begin{flushleft}}
\newcommand{\efl}{\end{flushleft}}
\newcommand{\bgnv}{\begin{verbatim}}
\newcommand{\ednv}{\end{verbatim}}

\newcommand{\myb}{\hspace{-0.1in}}

\newcommand{\myeq}{& \hspace{-0.1in} = & \hspace{-0.1in}}

\newcommand{\lb}{\nonumber \\}
\newcommand{\myarr}{\begin{array}{lll}}

\newcommand{\mygeq}{& \myb \geq & \myb}
\newcommand{\myg}{& \myb > & \myb}
\newcommand{\myleq}{& \myb \leq & \myb}
\newcommand{\myl}{& \myb < & \myb}
\newcommand{\bitem}{\begin{itemize}}
\newcommand{\eitem}{\end{itemize}}
\newcommand{\benum}{\begin{enumerate}}
\newcommand{\eenum}{\end{enumerate}}

\newcommand{\myhb}{\hspace{-0.3in}}
\newcommand{\myhf}{\hspace{0.3in}}
\newcommand{\myf}{\hspace{0.1in}}

\newcommand{\myskip}{\\ \vspace{-0.1in}}

\newcommand{\ER}{Erd$\ddot{\rm o}$s-R$\acute{\rm e}$nyi }

\def\QED{~\rule[-1pt]{5pt}{5pt}\par\medskip}
\newenvironment{proof}{{\bf Proof: \ }}{ \hfill \QED}

\newcommand{\change}[1]{{\color{black}{#1}}}

\begin{document}
%
\title{Effects of Degree Correlations in Interdependent 
	Security: Good or Bad?}


\author{Richard J. La\thanks{This work was 
supported in part by contracts 70NANB13H012,  
70NANB14H015 and 70NANB16H024 from National Institute
of Standards and Technology.} 
\thanks{Author is with the Department of Electrical \& 
Computer Engineering (ECE) and the Institute for Systems 
Research (ISR) at the University of Maryland, College Park.
E-mail: hyongla@umd.edu}
}

\maketitle

\begin{abstract}
We study the influence of {\em degree correlations} or 
{\em network mixing} in interdependent security. We model the 
interdependence in security among agents using a
dependence graph and employ 
a population game model to capture the interaction among 
{\em many} agents when they are {\em strategic} and have
various security measures they can choose to defend themselves.
The overall network security is measured by 
what we call the {\em average risk exposure} (ARE)
from neighbors, which is proportional to the total (expected)
number of attacks in the network.

We first show that there exists a unique pure-strategy Nash 
equilibrium of a population game. Then,
we prove that as the agents with larger
degrees in the dependence graph see higher
risks than those with smaller degrees, the overall network 
security deteriorates in that the ARE 
experienced by agents increases and there are more
attacks in the network. Finally, using this 
finding, we demonstrate that the effects of network mixing
on ARE depend on the {\em (cost) 
effectiveness} of security measures available to agents; 
if the security measures are not effective, increasing
assortativity of dependence graph results in 
higher ARE. On the other hand, if the
security measures are effective at fending off the 
damages and losses from attacks, 
increasing assortativity reduces the
ARE experienced by agents. 

\end{abstract}


\IEEEpeerreviewmaketitle


\section{Introduction}
	\label{sec:Introduction}

As many critical engineering systems, such as power grids, become 
more connected, there is a growing interest in understanding the 
security of large, complex networks in which security of many
comprising agents or subsystems is interdependent. This is
dubbed {\em interdependent security} (IDS) by Kunreuther and
Heal~\cite{KunHeal2003}. It arises naturally in many settings, 
and examples include cybersecurity \cite{BolotLelarge2008, 
Jiang2011, LelargeBolot2008, LelargeBolot2009}, 
cyber-physical systems security (e.g., power grids)
\cite{Bou-Harb}, 
epidemiology \cite{Pastor2005, Schneider2011}, 
financial networks and systems~\cite{Beale2011, Caccioli2011, 
Caccioli2012}, 
homeland security~\cite{HealKun2002, KunMichel2009}, and
supply chain and transportation system security (e.g., 
airline security)~\cite{Gkonis2010, 
KearnsOrtiz}.

The sizes and complexity of these systems as well as 
the number of participating agents introduce several 
major challenges to studying their reliability and 
security. This is especially the case
when they contain {\em many} 
individuals, organizations or (sub)systems that can 
make {\em local} security decisions {\em based on 
locally observable risks}.
Throughout the manuscript, we refer to 
these individuals, organizations or systems that make
own security decisions simply as {\em agents}.

First, in many cases, it is reasonable to assume that 
the agents are {\em rational}
or {\em strategic} and are only interested in their own
objectives with little or no regards for others. 
Therefore, a study of {\em static} settings in which
the agents make decisions without taking into account
the experienced risks may not be realistic.
Second, in IDS settings, the security of individual agents 
is interdependent, thereby causing the agents' security 
decisions to be {\em coupled} as a result of 
externalities produced by security measures they employ.
Furthermore, these externalities and the resulting
security risks seen by agents depend on the properties
of their dependence structure.  
Third, any attempt to model and study detailed interactions 
between {\em many} strategic agents suffers from the 
{\em curse of dimensionality}. Finally, while there are
some popular metrics used in the literature (e.g., 
global cascade probability), there is a {\em lack of 
standard metrics} on which security experts agree
for measuring network- or system-level security.

As mentioned above, the security of the systems in IDS settings 
depends on many system properties, including the properties
of interdependence in security among agents, which we model 
using a {\em dependence graph}. 
Although the effects of
some graph properties (e.g., degree distributions and
clustering~\cite{CoupLelarge2, La_CDC2014, La_TON}) have 
been recently studied in the literature, to the best of
our knowledge, the {\em influence of the 
degree correlations} in the dependence graph with
strategic agents has not
been examined before. The degree correlations, which are 
also known as {\em assortative mixing}, 
{\em (degree) assortativity} or 
{\em network mixing}, refer to the correlations in 
the degrees of end nodes of edges present in the graph. 

It has been shown \cite{Newman2002, Newman2003} that 
engineered networks, e.g., the Internet, tend to be 
disassortative, whereas social networks are typically
assortative. In other words, nodes in engineered 
systems tend to be connected 
to other nodes with {\em dissimilar} 
degrees, while those in social networks 
exhibit a tendency to be neighbors 
with other nodes with {\em similar} degrees. These
correlations in the degrees of end nodes in the dependence 
graph change the security risk experienced by agents 
from their neighbors based on their own degrees.
This is because the security investments chosen by 
agents with different degrees are likely to vary and 
some agents are more vulnerable to attacks
than others.
The goal of our study is to shed some light on how the 
{\em degree correlations} in the dependence
graph affect the security 
investments of strategic agents and, in doing so, 
{\em the overall system security}. 

While there have been some {\em numerical} studies on 
the influence of network mixing on the robustness
of networks in {\em static} settings 
(e.g., \cite{Newman2002, Newman2003}),
as we will discuss in Section~\ref{sec:Related},
there are two key differences between our study and
existing studies: 
(i) In our study, 
agents are {\em strategic} and can choose how much they wish
to invest in security in response to the security risks they 
experience. (ii) The security measures adopted by an agent 
(e.g., incoming traffic 
monitoring, anti-malware utility) produce positive 
externalities \cite{Varian_Microeconomics}
and {\em alter the security risks and threats 
seen by other agents} in the network, thereby influencing
their security investments.

It is shown that positive externalities produced by security 
measures on neighbors often lead to {\em free riding}
\cite{La_CDC2014, La_TON, LelargeBolot2009}; when some agents 
invest in security measures, positive 
externalities they generate curtail the risk experienced by 
other agents, thus reducing 
their incentive to protect themselves and invest in security.
Consequently, they cause {\em under-investments}
in security by strategic agents and social inefficiency
\cite{Zhao2009}. For this reason, the presence of externalities 
in IDS considerably complicates the analysis of the interactions 
among strategic agents. 

Let us illustrate these concepts with the help of the following
example. 

\noindent
$\diamond$ {\bf Spread of malware via emails:} 
When a user's device 
is infected by malware, it can scan the user's 
emails or the hard disk drive of the infected machine 
and send the 
user's personal or other confidential information to 
criminals interested in stealing, for instance, the user's 
identity (ID) or trade secrets. 
Moreover, the malware can browse the user's address book 
and either forward it to attackers or 
send out bogus emails, i.e., email spoofing, 
with a link or an attachment to those on the contact list. 
When a recipient clicks on the link or
opens the attachment, it too becomes infected.

In order to reduce the risks or threats from
malware, users can 
install an anti-malware utility on their 
devices. When a user adopts an anti-malware tool, 
not only does it reduce its own risk, but it
also lessens the risk to those 
on its address book for the reason stated above, 
in doing so protecting its friends to some degree. 
Therefore, it produces {\em positive externalities} 
for others \cite{ShapiroVarian, Varian_Microeconomics}. 
Interestingly, these positive externalities diminish 
the benefits of installing
anti-malware utilities for others, thus introducing
{\em negative network effects} for them.

\subsection{Summary and main contributions}

For mathematical tractability, we employ a 
{\em population game}~\cite{Sandholm} to model
the interactions among agents. This model is a 
generalization of the model used in our previous studies
that considered {\em neutral dependence graphs}
\cite{La_CDC2014, La_TON};
we assume a {\em continuous} action space, where 
an action represents the security investment
chosen by an agent with an understanding that the agent
selects the best combination of security measures 
subject to the budget constraint.

In order to measure the {\em global} network security 
and the {\em local} security experienced by individual 
agents, 
which is then utilized for choosing security investments, 
we adopt what we call the {\em average risk exposure} 
(ARE) from neighbors. 
While other global metrics, such as the 
probability of cascading failures/infections, 
have been adopted by existing studies, including our 
previous study~\cite{La_TON}, we argue that the ARE is a 
more natural and meaningful metric for our purpose
for the following reasons.

Since the agents can base their decisions only 
on {\em local} information or risks they can observe 
and assess, we need to model the local security 
risks they experience. First, we will show that the
ARE captures the {\em average security risks 
agents of varying degrees perceive from a neighbor},  
which allow them to approximate their total security 
risks from all neighbors.
Second, the agents are unlikely to have access
to the value of a commonly adopted global metric 
(e.g., cascade probability) as they lack 
global information, including
network topology and the security decisions of 
other agents. This makes such global metrics
unsuitable as information on which the agents
can act. In contrast, the ARE also serves as a 
{\em global} security metric because
it is proportional to {\em the total (expected)
number of attacks in the network}. For this 
reason, it provides us with a consistent metric 
for (a) measuring the global security and (b) 
capturing the local security information on which 
agents act, and enables us to compare the overall 
network security as we vary the properties of
dependence graph.

Our main contributions can be summarized as follows:
\myskip

{\bf S1.}
We show that there exists a unique (pure-strategy) 
Nash equilibrium (NE) of a population game under a 
mild technical condition. Then, we examine how the
assortativity of dependence graph changes the ARE 
at the unique NE as agents with varying degrees 
experience different risks from their neighbors due
to degree correlations. In particular, 
we prove that when the agents with larger
degrees in the dependence graph see higher
risks than those with smaller degrees, the overall 
network security deteriorates in that the ARE
experienced by agents increases and there are more 
attacks in the network. 

{\bf S2.} 
Making use of this 
finding, we demonstrate that the effects of network 
mixing on ARE depend on the {\em cost 
effectiveness} of security measures available to agents; 
if the security measures are not effective, increasing
assortativity of dependence graph results in 
higher ARE. On the other hand, if the
security measures are effective at lowering the 
damages and losses from attacks, increasing 
assortativity reduces the ARE experienced by agents. 

{\bf S3.}
Using numerical studies, we examine how the
cost effectiveness of security measures and the
sensitivity of ARE to the vulnerability of agents to
attacks shapes
the influence of assortativity. Numerical results 
suggest that as security measures improve and become
more effective at fending off attacks, the assortativity 
of dependence graph has greater effects on network 
security. Similarly, when ARE is more sensitive to 
agents' vulnerability to attacks, assortativity has 
stronger impact on equilibrium ARE. 
\myskip

As summarized in the following section, existing studies 
demonstrated that the 
assortativity of a network can significantly affect its
robustness and resilience, e.g., 
\cite{Newman2003, YenReiter2012, Zhao2009}.  Thus, 
understanding the effects of dependence graph properties
is important to (i) predicting the overall network- or
system-level security and (ii) devising sound policies. 

While we admit that our analysis is carried out using a 
simplified model, to the best our knowledge,  
our work here and in \cite{La_CDC2014, La_TON, La_TNSE}
is the first (analytical) study of how the network security 
is shaped by the properties
of dependence graph that governs the interdependence in 
security among strategic agents. Unlike our 
previous studies that assumed neutral dependence graphs, 
however, the focus of the current study is the impact of 
degree correlations in the dependence
graph on network security. As summarized earlier, 
incorporating the strategic nature of agents leads to 
somewhat unexpected and interesting observation that the
net influence of degree correlations is also determined by
the effectiveness of available security measures.

We believe that the {\em qualitative}
nature of our findings provides valuable insights into 
the behavior of strategic agents in IDS settings, which
we hope would be helpful in (i) understanding the pitfalls
in studying the security of complex systems and (ii)
designing better security policies and regulations. 
Finally, we emphasize that our goal is to understand 
{\em the effects of heterogeneous security risks 
experienced by agents based on their degrees} (due to 
degree correlations) on network security, 
as opposed to accurate modeling of 
assortativity observed in real networks. Thus, it is 
not our intent to develop a more accurate model of
dependence graphs with degree correlations. 
\myskip

The remainder of the paper is organized as follows:  
We provide a short survey of closely related literature
in Section~\ref{sec:Related}.
Section~\ref{sec:Model} describes the population game
model we adopt for our analysis, and Section
\ref{sec:ARE} introduces the security metric
we employ for comparing network security and explains
how we model the effects of degree correlations. Section
\ref{sec:Preliminaries} introduces some preliminary 
results we need for our main findings reported in 
Section~\ref{sec:MainResults}. Numerical results
are provided in Section~\ref{sec:Numerical}.
We conclude in Section~\ref{sec:Conclusion}.

\section{Related literature}
	\label{sec:Related}
	
There are existing studies on IDS, 
many of which employ a game theoretic
approach to model the strategic nature of agents 
(e.g., \cite{Jiang2011, KearnsOrtiz, KunHeal2003}). 
We refer an interested reader to a survey paper by Laszka et al.
[27] and references therein for a succinct discussion of these
and other related studies. In addition, many researchers
investigated the existence of assortativity in many different
types of networks, e.g., \cite{BaglerSinha2007, Newman2003, 
Pira2012, YenReiter2012}. Here, we only focus on studies that
examined the effects of assortativity in epidemics or 
security-related settings and summarize their main findings.

In \cite{Newman2002, Newman2003}, Newman studied 
network 
mixing in different types of networks, including biological
networks, engineered networks, and social networks. He 
first showed that while social networks in general exhibit 
assortativity, both engineered and biological networks 
tend to be disassortative. He then investigated how 
assortativity affects the phase transition in the emergence
of a giant component in random graphs as the average degree
of nodes increases. 

His findings revealed that stronger 
assortativity makes it easier for a giant component to
appear, but at the same time, the size of the giant
component tends to be smaller. Furthermore, breaking
up the giant component by removing a subset of nodes with
the highest degrees becomes more difficult when the network is
assortative; the numerical results suggest that the number
of nodes that must be removed from the network to split
up the giant component in an assortative network can be an 
order of magnitude larger than that of a disassortative
network. 

He argued that these findings have following important 
implications. First, preventing an outbreak of a disease 
via vaccination of high-degree individuals can be problematic 
because social networks exhibit assortativity and 
the cluster of high-degree nodes could serve as a reservoir
of disease. However, for the same reason, an epidemic
will likely be limited to a smaller portion of 
population if an outbreak does occur. On the
other hand, improving the resilience of engineered 
networks, such as the Internet, which show disassortativity
becomes more challenging as disassortative networks 
are more susceptible to coordinated attacks that 
target high-degree nodes in the networks. 

In \cite{Boguna2003a, Boguna2003b}, Bogu$\tilde{\rm n}
\acute{\rm a}$ et al. studied the (existence) of epidemic 
threshold in scale-free networks, using the popular
susceptible-infected-susceptible model. The key 
finding of the study was that the degree correlations do 
not significantly affect the existence of epidemic 
threshold as long as the degree correlations are limited 
to immediate neighbors. Instead, the (lack of) the 
existence of threshold is shaped by the divergence of 
the second moment of node degrees when the power law
exponent lies in the interval (2, 3]. 

Another related study by Zhou et al.~\cite{Zhou2012}
investigated the influence of assortativity on the
robustness of {\em interdependent networks} with the 
help of independent failure model, using both 
\ER networks and scale-free networks. Their
main finding suggests that increasing assortativity
leads to deteriorating robustness of interdependent
network; as a network becomes more assortative, 
the initial number of nodes that need to be removed
in order to break up the giant component in the 
network drops. This indicates that it is easier to 
break up the giant component in an interdependent
network by eliminating randomly chosen nodes. 

In a more concrete cybersecurity application, 
Yen and Reiter
\cite{YenReiter2012} studied how the assortativity 
of botnets influences the performance of takedown
strategies. They first demonstrated that botnets
exhibit high assortativity and attributed this
in part to the working of botnets. 
Secondly, they showed that some of well studied 
takedown strategies, in particular uniform 
takedown and degree-based takedown strategies, 
are far less effective as the botnets become
more assortative. This finding suggests that previous
studies carried out with neutral botnets may be 
inaccurate and incorrectly portray a more optimistic
picture. Finally, they also considered
other alternative takedown strategies that 
take into account clustering coefficients and 
closeness centrality and showed that a similar
trend continues.

We note that these studies do not take into consideration 
the strategic nature of individual agents that can make
their own security decisions, which is 
natural in many settings of interest. 
Our study considers strategic agents that determine 
their security investments in response to the security
risks they observe. 
In addition, rather than focusing on giant components
in networks and possible cascades of infection, 
we analyze the (local) network security experienced by 
individual agents {\em as a result of their security 
decisions} at equilibria.

We studied related problems in 
\cite{La_CDC2014, 
La_TON, La_TNSE} under {\em neutral} dependence graphs. 
In \cite{La_TNSE}, we investigated (i) how we could 
improve the overall (network) security by 
{\em internalizing} the externalities 
produced by the security measures adopted by agents
and (ii) how the sensitivity of network security to 
agents' security investments influences the penalties 
or taxes that need to be imposed on the agents to
internalize externalities. 
Moreover, we showed \cite{La_CDC2014, La_TNSE}
that as the security of agents 
gets more interdependent in that their degrees
in the dependence graph become larger (with respect 
to the usual stochastic order \cite{SO}), the security 
experienced by agents whose degrees remain fixed
improves in that the number of attacks they suffer 
goes down. Thus, this finding tells us, to some extent,
how the {\em degree distribution} in the dependence
graph affects the network security.

In \cite{La_TON}, we considered a simple model 
where agents can choose from three possible actions: i) 
invest in security, ii) purchase security insurance 
to transfer (some of) risks, and iii) take no actions.
Using this model, we carried out {\em numerical studies} 
that examined how the degree distribution of dependence
graph affects 
the {\em cascade probability}. Our study demonstrated 
that as the interdependence in security rises, so does 
the probability of cascade. Moreover, we derived
an upper bound on the price of anarchy, i.e., the 
ratio of the social cost at the Nash equilibrium to
that of the social optimum, which is a linear function
of the average node degree.

We point out that none of the above studies,
including our own studies, 
investigated the role of network mixing in IDS settings 
with strategic agents and no analytical findings have 
been reported. A key difference between our study
in \cite{La_CDC2014, La_TNSE} and the current 
study is the following: our previous study focused on how 
varying degree distributions influence the network security
in a neutral dependence graph. The current study, on the
other hand, considers a {\em fixed degree distribution} 
and examines how differing security risks seen by agents
based on their own degrees (due to degree correlations), 
shape the resulting network security. 
Some of our preliminary results have been 
reported in \cite{La_Globecom}. It, however, employs a 
simpler, hence more restrictive model to facilitate the 
analysis.

\section{Model}
\label{sec:Model}

We capture the interdependence in security among the 
agents 
using an undirected graph, which we call the {\em dependence
graph}. A node or vertex in the graph corresponds to an 
agent (e.g., an individual or 
organization),
and an undirected edge between nodes $n_1$ and $n_2$ implies 
interdependence of their security. We interpret an undirected edge
as two directed edges pointing in the opposite directions
with an understanding that a directed edge from node $n_1$ to node
$n_2$ indicates that the security of node $n_1$ affects that of
node $n_2$ in the manner we explain shortly. When there is
an edge between two nodes, we say that they are {\em immediate}
or {\em one-hop} neighbors or, simply, neighbors when it is clear.

We model the interaction among agents as a {\em 
noncooperative game}, in which players are the 
agents.\footnote{We will use the words {\em agents}, 
{\em nodes} and {\em players} interchangeably hereafter.} 
This is reasonable because, in many cases, it may be difficult 
for agents to cooperate with each other and take coordinated 
countermeasures 
to attacks. In addition, even if they could coordinate their 
actions, they would be unlikely to do so when there are no clear 
incentives for coordination. 

We are interested in scenarios where the number of agents is large. 
Unfortunately, modeling detailed {\em microscale} interactions 
among many 
agents in a large network and analyzing ensuing games is difficult;
the number of possible strategy profiles typically increases
exponentially with the number of players and finding the NEs of 
noncooperative games is often challenging even with a moderate 
number of players.

The notation we adopt throughout the paper is listed in 
Table~\ref{tab:notation}.

\begin{table}[h]
\begin{center}
 
\begin{tabular}{l|l}
\hline 
$C(\bx, d, a, \bs)$ & Cost of an agent with degree $d$ 	
	playing action $a$ at \\
& \myf social state $\bx$ \\
$\cD$ & Set of agent degrees or populations \\
& \myf ($\cD = \{1, 2, \ldots, D_{\max}\}$) \\
$D_{\max}$ & Maximum degree among agents or the number 
	of \\
& \myf populations, i.e., $D_{\max} = |\cD|$ \\
$\cI$ & (Pure) action space ($\cI = [I_{\min}, 
	I_{\max}]$) \\ 
$I^{\opt}(r)$ & Optimal security investment of an agent 
	facing $r$ \\
& \myf expected attacks \\
$L$ & Average loss from a single infection \\
$\cP_{\cI}$ &  Set of probability distributions over 
	$\cI$ \\
$\cX$ & Cartesian product $\cP_{\cI}^{D_{\max}}$ \\
$d_{\avg}$ or $d_{\avg}(\bs)$ 
	& Average or mean degree of agents \\
& \myf ($d_{\avg}(\bs) = \sum_{d \in \cD} d \cdot 
	f_d(\bs)$) \\
$e_{\avg}(\bx, \bs)$
	& Average risk exposure at social state $\bx$ \\
$e_d(\bx, \bs)$
	& Risk exposure of pop. $d$ at social state
	$\bx$ \\
$f_d$ or $f_d(\bs)$ & Fraction of agents with degree 
	$d$ \\
& \myf ($f_d(\bs) = s_d / \sum_{d' \in \cD} s_{d'}$) \\
$\bg$ & Mixing vector ($\bg = (g_d; \ d \in \cD)$) \\
$p(a)$ & Infection prob. of an agent investing $a$
	in security \\
$p^\star(r)$ & Infection prob. of an agent facing $r$
	expected attacks \\
& \myf and investing $I^{\opt}(r)$ in security \\ 
& \myf ($p^\star(r) = p(I^{\opt}(r))$) \\
$p_{d, \avg}(\bx)$ & Average infection prob. of 
	population $d$ at social \\
& \myf state $\bx$ \\
$\bs$ & Pop. size vector ($\bs = (s_d; \ d \in \cD)$) \\
$s_d$ & Size of pop. $d \in \cD$ \\
$w_d$ or $w_d(\bs)$ & Weighted fraction of agents with 
	degree $d$ \\
& \myf \Big($w_d(\bs) = \frac{d \cdot s_d}{\sum_{d' \in \cD}
	d' \cdot s_{d'}} = \frac{d \cdot f_d(\bs)}
		{d_{\avg}(\bs)}$ \Big)  \\
$\bx_d$ & Pop. state of pop. $d$ \\
$\bx$ & Social state ($\bx = (\bx_d; \ d \in \cD)$) \\
$\beta_{IA}$ & Prob. of indirect attack on a neighbor by an
	infected \\
& \myf agent \\
$\tau_{A}$ & Prob. that an agent experiences a direct attack \\
\hline
\end{tabular}

\caption{Notation ({\rm pop. $=$ population, prob. $=$ probability}).} 
\label{tab:notation}
\end{center}
\vspace{-0.25in}
\end{table}

\subsection{Population game model}

For analytical tractability, we adopt a population game
with a continuous action space to model the interaction 
among the agents~\cite{Sandholm}. As stated earlier, 
the (local) network security is captured using ARE from  
neighbors. As we explain in Section~\ref{subsec:RE-ARE},
the ARE is proportional to the 
total (expected) number of attacks that propagate from the 
victims of successful attacks to their neighbors in the
network and can be viewed as a measure of global
network security. 

We assume that the maximum degree among all agents in the
dependence graph is $D_{\max} < \infty$.  For each $d \in 
\{1, 2, \ldots, D_{\max}\} =: \cD$, population $d$ 
consists of all agents with common degree $d$. Let $s_d$ 
denote the {\em size} or {\em mass} of population $d$, and 
the population size vector ${\bf s} := \big( s_d; \ d \in 
{\cal D} \big)$ tells us the sizes of populations with varying 
degrees.\footnote{Throughout the paper, all vectors
are assumed to be column vectors.}  

We find it convenient to define ${\bf f}({\bf s}) := (f_d({\bf s}); 
\ d \in \cD)$, where $f_d({\bf s}) = s_d / \sum_{d' \in \cD} s_{d'}$
is the {\em fraction} of agents with 
degree $d$ in the dependence graph. 
Given a population size vector ${\bf s}$, 
we denote the average degree of
agents by $d_{\avg}(\bs) := \sum_{d \in \cD} d \cdot f_d({\bf s})$. 
When there is no confusion, we simply denote ${\bf f}({\bf s})$ and
$d_{\avg}({\bf s})$ by ${\bf f}$ and $d_{\avg}$, respectively. 
\myskip

$\bullet$ {\bf Population state and social state --}
All agents have the same action space $\cI = [I_{\min}, 
I_{\max}] \subset \R_+ := [0, \infty)$, 
where $I_{\min} < I_{\max} < \infty$. 
A (pure) action taken by an agent represents the security 
investment made by the agent. We denote the
set of probability distributions over $\cI$ by ${\cal P}_\cI$. 

The {\em population state} of population $d$ is given 
by ${\bf x}_d 
\in {\cal P}_\cI$. In other words, given any (Borel) subset
${\cal S} \subseteq \cI$, $\bx_d({\cal S})$ tells us the {\em 
fraction} of population $d$ whose security investment lies in 
${\cal S}$.  The {\em social state}, denoted by $\bx = (\bx_d; \ 
d \in \cD) \in \cX := {\cal P}_\cI^{D_{\max}}$, specifies
the actions chosen by all agents.  
\myskip

$\bullet$ {\bf Two types of attacks --} 
In order to understand how the degree correlations of dependence 
graph affect the security
investments of the agents and overall network security, we 
model two different types of attacks agents suffer from 
-- {\em direct} and {\em indirect} attacks. 
While the first type of attacks is independent
of the dependence graph, the latter depends on it, allowing 
us to capture the {\em externalities} produced by agents' 
security choices. 

{\em a) Direct attacks: }
We assume that malicious attackers launch attacks 
on the agents, which we call {\em direct} attacks. 
While our model can be easily
modified to handle a scenario in which an agent can suffer
more than one direct attack from different attackers by
modifying the cost function, here 
we assume that an agent experiences at most one direct 
attack and the probability of bearing a direct attack
is $\tau_A$, independently of other agents. 

When an agent experiences a direct attack, its cost 
depends on its security investment; when an
agent adopts action $a \in \cI$, it is infected with 
probability $p(a) \in [0, 1]$. Also, each time an agent is
infected, it incurs on the average a cost of $L$. Hence, the 
expected cost or 
loss of an agent from a single attack is $L(a) := 
L \cdot p(a)$ when investing $a$ in security. 

It is shown~\cite{Bary} that, under some technical assumptions, 
the security breach probability or probability of loss is a 
{\em log-convex} (hence, strictly convex) decreasing function of 
the investments. Based on this finding, we introduce the following
assumption on the infection probability $p(a)$, $a \in \cI$. 
\myskip

\begin{assm} 	\label{assm:pa}
The infection probability $p: \cI \to [0, 1]$ 
is {\em continuous, decreasing and strictly convex}. Moreover, 
it is continuously differentiable on ${\rm int}(\cI)
= (I_{\min}, I_{\max})$. 
\myskip
\end{assm}

{\em b) Indirect attacks: }
Besides the direct attacks by the attackers, an agent may
also experience {\em indirect} attacks from its neighbors that have
sustained successful attacks and are infected.  We assume 
that an infected agent will unwittingly participate in indirect 
attacks on its neighbors, each of which is attacked 
with probability $\beta_{IA} \in (0, 1]$
independently of each other. When an agent investing $a$
in security suffers an indirect attack, it is infected with the same 
infection probability $p(a)$. 

We call $\beta_{IA}$ indirect attack probability (IAP).
It affects the {\em local} spreading behavior. Unfortunately, 
the dynamics of infection propagation depend on the details of
underlying dependence graph, which are difficult to obtain
or model faithfully. In order to skirt this difficulty, instead of
attempting to model the detailed dynamics of infection 
transmissions between agents, we abstract out the {\em security 
risks} seen by the agents using {\em the expected number of attacks 
an agent sees from its neighbors}.
However, to capture the effects of network mixing, we
allow agents of varying degrees to experience different risks from 
their neighbors as explained below and in Section~\ref{sec:ARE}. 
\myskip

$\bullet$ {\bf Cost function --} 
The cost function of the game is determined by a 
function $C: \cX \times \cD \times \cI \times \R_+^{D_{\max}} 
\to \R$. The 
interpretation is that, when the population size 
vector is $\bs$ and the social state is $\bx$, the cost of 
an agent with degree $d$ (hence, from population $d$) 
playing action $a \in \cI$ (thus, investing $a$ in security) is 
equal to $C(\bx, d, a, \bs)$. As we will show below, in addition 
to the cost of security investments, our cost function 
also reflects the (expected) losses from attacks. 

Given a social state $\bx \in \cX$, let $e_d(\bx, \bs)$ 
denote the average number of indirect attacks an 
agent with degree $d \in \cD$ sees from a {\em single} 
neighbor. Hence, the average number of indirect 
attacks experienced by agents of degree $d$ would be
$d \cdot e_d(\bx, \bs)$. One natural metric for
the security risk seen by agents is the number of 
attacks they expect to see. Hence, $e_d(\bx, \bs)$
captures the {\em security risk per neighbor} 
observed by agents of degree $d$. 
We call $e_d(\bx, \bs)$ the {\em risk exposure} (RE) 
for population $d$ at social state $\bx$. Since 
we are interested in understanding how network mixing 
affects the agents' security investments, it is 
necessary to allow the RE to vary from one population 
to another, i.e., $e_{d}(\bx, \bs)$
and $e_{d'}(\bx,\bs)$ can differ if $d \neq d'$. 

Before we proceed, let us comment on the key 
difference between the current model and that of our 
earlier work~\cite{La_TNSE}, which only considers 
{\em neutral} dependence graphs with no degree correlations. 
When the underlying dependence graph is neutral, the 
degree distribution of neighbors does not depend on the
degree of the agent under consideration and 
the risk exposure is identical for populations, 
i.e., $e_d(\bx, \bs) = e_{d'}(\bx, \bs)$ for all 
$d, d' \in \cD$. 
As a result, both the model and the analysis become 
much simpler. 

We assume that the costs of an agent due to multiple 
infections are additive. Hence, the expected cost of an 
agent with degree $d$ from indirect attacks is 
proportional to its degree and RE $e_d(\bx, \bs)$. 
\change{The additivity of costs is reasonable in many 
scenarios, 
including the earlier example of malware propagation; 
each time a user is infected by different malware
(e.g., ransomware)
or its ID is stolen, the user will need to spend time 
and incur expenses to deal with the problem.  
Similarly, every time a corporate network
is breached, besides any financial losses or legal
expenses, the network operator will need to assess
the damages and take corrective measures.} 

Based on this assumption, we adopt the following cost function 
for our population game: for a given social state ${\bf x} 
\in \cX$, the cost of an agent with degree $d$ 
investing $a$ in security is equal to 
\beqa
\myhb C({\bf x}, d, a, \bs) 
\myeq \left( \tau_A + d \cdot e_d({\bf x}, \bs) \right) L(a) 
	+ a. 
	\label{eq:Cost}
\eeqa
Note that $\tau_A + d \cdot e_d(\bx, \bs)$ is the 
total number of both direct and indirect attacks an agent  
of degree $d$ expects. Hence, the first term on the 
right-hand side of (\ref{eq:Cost}) is the total expected 
loss due to infections. 

From now on, we take the viewpoint that the 
agents use their expected number of attacks given 
by $\tau_A + d \cdot e_d(\bx, \bs)$ as their perceived
security risks at social state $\bx$. Based on these
observed risks, they decide their security investments 
to minimize their cost given in (\ref{eq:Cost}). 
\myskip

$\bullet$ {\bf Nash equilibria --}
We focus on the NEs of population games as an 
approximation to agents' behavior in practice. 
For every $d \in \cD$, define a mapping $\cI^{\opt}_d: \cX
\to \mathcal{B}(\cI)$, where $\mathcal{B}(\cI)$ is the 
set of (Borel) subsets of $\cI$ and 
\beqan
\cI^{\opt}_d(\bx) 
:= \left\{ a \in \cI \ \big| \ C(\bx, d , a, \bs) 
	= \inf_{a' \in \cI} C(\bx, d, a', \bs) \right\}. 
\eeqan

\begin{defn}
A social state ${\bf x}^\star$ is an NE if  
$\bx^\star_d(\cI^{\opt}_d(\bx^\star)) = 1$ for all
$d \in \cD$. 
\myskip
\end{defn}

Clearly, our model does not require that all agents from a
population adopt the same action in general. However, 
we are often interested in cases in 
which the social state is degenerate, i.e., all agents with 
the same degree adopt the same action. In this case, we denote
the action chosen by population $d \in \cD$ by $a_d$, and 
refer to $\ba := (a_d; \ d \in \cD) \in \cI^{D_{\max}}$ as 
a {\em pure strategy profile}. 

With a little abuse of notation, we denote the RE of 
population $d \in \cD$ when a 
pure strategy profile $\ba$ is employed by 
$e_d(\ba, \bs)$. 
\myskip

\begin{defn}	\label{defn:PSNE}
A pure strategy profile $\ba^\star \in \cI^{D_{\max}}$ is said
to be a pure-strategy NE if, for all 
$d \in \cD$,
\beqan 
a^\star_d \in \arg \min_{a \in \cI} \big( (\tau_A + d \cdot 
	e_d(\ba^\star, \bs)) L(a) + a \big). 
\eeqan
In other words, every agent in a population adopts the 
same best response. 
\end{defn}

\section{Average risk exposure and the effects
	of network mixing}
	\label{sec:ARE}

In this section, we first define the security 
metric we adopt to measure the (global) network 
security, namely ARE, and describe how we 
estimate it. 
Then, we lay out how we model the {\em 
the influence of degree correlations} on the 
average security risks experienced by agents of
varying degrees (measured by the expected
number of attacks) via the REs $e_d(\bx, 
\bs)$, $d \in \cD$.

\subsection{Average risk exposure} 
	\label{subsec:RE-ARE}

As mentioned in Section~\ref{sec:Introduction}, 
we use a metric we call ARE to measure and 
compare the network security as we study the 
impact of degree correlations. The ARE is defined 
to be the (expected) {\em total number of 
indirect attacks} experienced by all agents 
divided by the number of directed edges in the
dependence graph. Since $e_d(\bx, \bs)$ is the
number of indirect attacks an agent of degree $d$
expects from a single neighbor at social state
$\bx$, its expected total number of indirect 
attacks is $d \cdot e_d(\bx, \bs)$. Therefore, 
the expected total number of indirect attacks in
the network is equal to $\sum_{d \in \cD} \left( 
s_d \times d \cdot e_d(\bx, \bs) \right)$, and
the ARE is given by 
\beqa 
\myhb e_{\avg}(\bx, \bs) 
\myeq \frac{ \sum_{d \in \cD}  s_d \cdot d \cdot 
	e_d(\bx, \bs)}{ \sum_{d \in \cD} s_d \cdot d} 
	\label{eq:ARE0} \\
\myeq \sum_{d \in \cD} w_d(\bs) \cdot e_d(\bx, \bs), 
	\label{eq:ARE}
\eeqa
where $w_d(\bs) := d \cdot f_d(\bs) / d_{\avg}(\bs)$, 
$d \in \cD$.

Since it is by definition proportional 
to the expected {\em total number of indirect 
attacks} in the network (for fixed degrees in 
the network), the ARE can 
be considered a {\em global} metric for network 
security which measures the {\em aggregate security 
risks to all 
agents} in the form of attacks from neighbors.
In the rest of the paper, we take this viewpoint and 
study how network mixing influences the network 
security measured by ARE.

While the definition of ARE is simple
and intuitive, it does not provide a means of 
computing ARE unless we already know the REs
$e_d(\bx, \bs)$ for all populations. 
Therefore, we need a way to estimate it.
Unfortunately, computing the ARE exactly starting 
with its definition suffers from several major
technical difficulties; it depends on the 
{\em detailed properties} of both the dependence 
graph and the dynamics of infection propagation 
among agents. Modeling these accurately is 
difficult, if possible at all. More importantly, 
such detailed models in general do not yield
to mathematical analysis. For these reasons, we 
seek to approximate ARE. 
\myskip

\subsubsection{Approximation of ARE}

In order to approximate the ARE (and the REs), 
we base our model on the following observation:
all {\em indirect} attacks begin with the {\em first-hop}
indirect attacks on the immediate neighbors by {\em the 
victims of successful direct attacks}. Thus, it 
is reasonable to assume that the {\em total} number
of indirect attacks in the network 
increases with the number of
the {\em first-hop} indirect attacks, each of which 
can initiate a chain of indirect attacks thereafter.

Let 
\beqan
p_{d, \avg}(\bx) := \int_{\cI}
	p(a) \ \bx_{d}(da)
\eeqan
be the probability that a randomly selected agent 
of degree $d$ will suffer an infection from a
{\em single} attack at social state $\bx$. 
The expected number of agents with degree $d$ 
which will fall
victims to {\em direct} attacks is 
$\tau_A \cdot 
s_d \cdot p_{d, \avg}({\bf x})$, and each infected
agent of degree $d$ will attempt to transmit the
infection to each of its $d$ neighbors with IAP
$\beta_{IA}$. Thus,  
the expected total number of first-hop
indirect attacks by the victims infected by
direct attacks is equal to $\tau_A \cdot 
\beta_{IA} \ 
\sum_{d \in {\cal D}} \left( d \cdot s_d 
\cdot p_{d, \avg}({\bf x}) \right)$.

Based on this argument, we approximate the 
ARE as a strictly increasing function of $\tau_A
\cdot \beta_{IA} 
\ \sum_{d \in \cD} (d \cdot s_d \cdot 
p_{d, \avg}(\bx))$. But, 
we first normalize it by 
the total population size 
and work with the expected number of
first-hop indirect attacks {\em per agent}, i.e., 
$\tau_A \cdot \beta_{IA} \
\sum_{d \in {\cal D}} \left( d \cdot s_d 
\cdot p_{d, \avg}({\bf x}) \right) / 
\sum_{d \in \cD} s_d = \tau_A \cdot 
\beta_{IA} \ \sum_{d \in {\cal D}} 
\left( d \cdot f_d(\bs) \cdot p_{d, \avg}({\bf x}) 
\right)$, where the equality follows from the 
definition of $\bff(\bs)$. In summary, we estimate 
the ARE using
\beqa
e_{\avg}(\bx, \bs)
\myeq \Theta\left( \sum_{d \in \cD} d \ f_d(\bs) 
	\ p_{d, \avg}(\bx) \right)
	\label{eq:ARE1}
\eeqa
for some strictly increasing function $\Theta: 
\R_+ \to \R_+$, which we assume factors in both 
$\tau_A$ and $\beta_{IA}$. A simple example of function
$\Theta$ is a linear function, i.e.,  
\beqa
e_{\avg}(\bx, \bs)
\myeq K \sum_{d \in \cD} d \ f_d(\bs) \ 
	p_{d, \avg}(\bx)
		\label{eq:ARE2}
\eeqa
for some $K > 0$.
The exact form of the function $\Theta$ will depend 
on many factors, including the detailed dynamics
of infection propagation, direct attack probability 
$\tau_A$, IAP $\beta_{IA}$, 
and the timeliness of deployed remedies (e.g., security 
patches) to stop the spread of infection.

We impose the following assumption on the function $\Theta$.
\myskip

\begin{assm}	\label{assm:theta}
The function $\Theta: \R_+ \to \R_+$ is continuous and 
strictly increasing. Furthermore, it is continuously 
differentiable over $\R_{++} := (0, \infty)$. 
\myskip
\end{assm}
The first part of the assumption is natural as argued above. 
While the latter part (i.e., continuous differentiability) 
is introduced for convenience to facilitate our analysis, 
we feel that it is reasonable; \change{recall that the
ARE is proportional to the expected total number
of indirect attacks in the network, and multi-hop 
indirect attacks can be viewed as offsprings of one-hop
indirect attacks, starting with the victims of direct
attacks. For this reason, 
in practice, we expect the average security risk 
measured by ARE to be a `smooth' 
function of the expected number of one-hop indirect 
attacks.} 
\myskip

\subsubsection{Alternate expression of ARE} 

Before we proceed, let us provide an alternate 
expression of ARE, which helps us highlight two distinct
sources that influence the ARE and isolate the one of
interest to us. 
To this end, let us define a 
mapping $\Phi: \R_+^{D_{\max}} 
\times \R_+ \to \R_+$ with $\Phi(\bs, r) = 
\Theta(d_{\avg}(\bs) \ r)$. From the definition 
of $\bw(\bs)$, we have the following equality.  
\beqan
\sum_{d \in \cD} d \ f_d(\bs) \ p_{d, \avg}(\bx)
\myeq d_{\avg}(\bs) \sum_{d \in \cD} w_d(\bs) 
	\ p_{d, \avg}(\bx)
\eeqan
As explained in \cite{La_TON, La_TNSE}, $\sum_{d \in 
\cD} w_d(\bs) \ p_{d, \avg}(\bx)$ is the probability 
that an end node of a randomly selected edge in the 
dependence graph is vulnerable to an attack,\footnote{This 
sampling technique is called {\em sampling by random 
edge selection} \cite{LeskovecFaloutsos}.} i.e., 
it becomes infected when attacked, at social
state $\bx$. Therefore, it captures on the average
how vulnerable neighboring agents are to indirect attacks 
and, hence, serves as an indicator of 
how easily an infection might transmit from one 
agent to another.

Using the definition of the mapping $\Phi$, the 
ARE can be rewritten as
\beqa
e_{\avg}(\bx, \bs) 
\myeq \Theta\left( d_{\avg}(\bs) \sum_{d \in \cD} 
	w_d(\bs) \ p_{d, \avg}(\bx) \right)
		\label{eq:ARE3} \\
\myeq \Phi\left(\bs, 
\sum_{d \in \cD} w_d(\bs) \ p_{d, \avg}(\bx) \right).
	\nonumber
\eeqa
From (\ref{eq:ARE3}), it is obvious that the ARE
depends on two measures that capture the ease with 
which an infection can spread through the network: 
(a) the average degree of agents, $d_{\avg}(\bs)$, 
indicates on the average how many other agents an 
infected agent could potentially infect, and (b)  
$\sum_{d \in \cD} w_d(\bs) \ p_{d, \avg}(\bx)$
tells us how vulnerable neighboring agents are in 
general. 

The first argument of $\Phi$ depends only on 
the dependence graph and is beyond the control of 
agents. Moreover, in our study, we assume that the 
population sizes $\bs$, hence the 
average degree $d_{\avg}(\bs)$, are {\em fixed} and 
study the influence of degree correlations. On the 
other hand, the second argument is a function of the 
social state $\bx$ chosen by agents. Thus, it 
incorporates the effects of degree correlations that 
induce {\em heterogeneous} REs seen by agents of 
varying degrees and, as a result, alter the
equilibrium ARE (and REs) by affecting their
security investments.

\subsection{The effects of network mixing}

The expression in (\ref{eq:ARE0}) tells us
how the REs shape the ARE. Another way of putting 
this is that, once the agents choose the social state
$\bx$ and the REs are fixed for all populations, we can 
compute the ARE and then {\em infer} the relations 
between the ARE $e_{\avg}(\bx, \bs)$ and individual 
REs $e_d(\bx,\bs)$, $d \in \cD$. These relations 
reveal {\em how the underlying degree correlations 
bias the REs} at the social state $\bx$ 
(relative to a 
neutral dependence graph under which $e_d(\bx, \bs) 
= e_{\avg}(\bx, \bs)$ for all $d \in \cD$ and 
all $\bx \in \cX$). Therefore, they summarize the 
net effects of degree correlations on security risks
experienced by agents based on their degrees.

In this paper, we assume that these relations 
are approximately linear. 
In other words, for every $d \in \cD$, there exists
some $g_d > 0$ such that $e_d(\bx, \bs) 
= g_d \cdot e_{\avg}(\bx, \bs)$ 
for all $\bx \in \cX$. The case with 
$g_d = 1$ for all $d \in \cD$ corresponds to 
the {\em neutral} dependence graph because $e_d(\bx, 
\bs) = e_{\avg}(\bx, \bs)$ for all $d \in \cD$, 
and agents see similar
risks from their neighbors regardless of their 
own degrees.

Obviously, this is a simplifying 
assumption and might not hold in practice. 
However, we feel that it is a reasonable first-order
approximation for {\em local} analysis around neutral 
dependence graphs at the NEs, 
which is the main focus of this paper 
(Theorem~\ref{thm:1} in Section~\ref{sec:MainResults}), 
and allows us to tackle otherwise a very difficult 
problem of understanding how different REs experienced 
by agents with varying degrees shape their
security investments and resulting network security. 

We refer to ${\bf g} := (g_d; d \in \cD)$ as 
a {\em mixing vector}. It models a {\em 
bias} or {\em skewness} in the average risk posed by 
neighbors to agents with varying degrees, which is 
caused by degree correlations. However, it does not 
correspond to any existing measure of assortativity, 
such as assortativity coefficient (which is Pearson
correlation coefficient). 
In this sense, we are
primarily concerned with capturing the {\em net effects 
of degree correlations} seen by agents with different
degrees, without having to worry about accurate 
modeling or measuring of assortativity itself.
 
For example, suppose that (i) agents with smaller degrees 
do not have a strong incentive to invest in security and 
fall victim to attacks more often than those with
larger degrees and (ii) the dependence graph exhibits
disassortativity (hence, agents with high degrees are 
more likely to be connected to agents of small degrees). 
Then, $g_d$ would be greater than one for large $d$ 
because they would see larger risks from their neighbors
with small degrees. Similarly, $g_d$ would be  
less than one for small $d$ because agents with small
degrees would be more likely to 
have neighbors with large degrees, which 
would pose lower risks.

\section{Preliminaries}
	\label{sec:Preliminaries}

From (\ref{eq:Cost}) and (\ref{eq:ARE1}), for a fixed 
mixing vector ${\bf g}$, the cost function is identical 
for two population size vectors $\bs^1$ and $\bs^2$ 
with the same node degree distribution, 
i.e., ${\bf f}(\bs^1) = {\bf f}(\bs^2)$. 
This scale invariance property of the cost function 
implies that the set of NEs is identical for both 
population size vectors. As a result, it suffices to 
study the NEs for population size vectors whose sum is 
equal to one, i.e., $\sum_{d \in \cD} s_d = 1$. For 
this reason, without loss of generality we impose the
following assumption in the remainder of the paper. 
\myskip

\begin{assm} 	\label{assm:0}
The population size vectors are normalized so that
the total population size is equal to one. 
\myskip
\end{assm}
Keep in mind that, under Assumption~\ref{assm:0}, the 
node degree distribution ${\bf f}(\bs)$ is equal to the
population size vector ${\bf s}$, i.e., ${\bf f}(\bs) = \bs$. 

Let us discuss a few observations that will help us prove the 
main results. 
\myskip

\noindent $\bullet$ 
{\bf Infection probability at optimal investments:}
For each $r \in \R_+$, let $I^{\opt}(r)$
be the set of optimal investments for an agent when
its security risk (measured by 
{\em the number of attacks it expects}) is $r$. In 
other words, 
\beqan
I^{\opt}(r) 
\myeq \arg \min_{a \in \cI} \big( r \ L(a)  
	+ a \big). 
\eeqan
Under Assumption~\ref{assm:pa}, one can show that the 
optimal investment is unique, i.e., $I^{\opt}(r)$ is a 
singleton for all $r \in \R_+$. Hence, we can view 
$I^{\opt} : \R_+ \to \cI$ 
as a mapping that tells us the optimal investment that will
be chosen by an agent as a function of the number of attacks 
it expects. This in turn implies that, at an NE $\bx^\star$, 
the population 
state $\bx^\star_d$ is concentrated on a single point, i.e., 
$\bx^\star_d( \{I^{\opt}( \tau_A + d  \cdot e_d(\bx^\star, 
\bs) )\} ) = 1$ for all $d \in \cD$. 

Define 
\beqan
r_{\min} 
= \sup\{ r \in \R_+ \ | \ I^{\opt}(r) = I_{\min} \}
\eeqan
and 
\beqan
r_{\max}
= \inf\{ r \in \R_+ \ | \ I^{\opt}(r) = I_{\max} \}. 
\eeqan
From their definitions, $r_{\min}$ (resp. $r_{\max})$
is the largest number of attacks (resp. the smallest
number of attacks) experienced by an agent, for which
the optimal investment is $I_{\min}$ (resp. $I_{\max}$).
Then, $I^{\opt}(r)$ is nondecreasing in $r$. Moreover, 
it is strictly increasing over [$r_{\min}, r_{\max}]$. 

Let the mapping $p^\star: \R_+ \to [0, 1]$ be the 
composition of $p: \cI \to [0, 1]$ and $I^{\opt}: \R_+ 
\to \cI$, i.e., $p^\star(r) = p\left( I^{\opt}(r) \right)$. 
The following corollary is an immediate consequence 
of (i) the assumption that $p$ is decreasing and (ii) the 
earlier observation that $I^{\opt}$ is nondecreasing (and 
strictly increasing over [$r_{\min}, r_{\max}$]). 
\myskip

\begin{coro}	\label{coro:OptInfProb}
The mapping $p^\star$ is nonincreasing. Furthermore, 
it is strictly decreasing over $[r_{\min}, r_{\max}]$.
\myskip
\end{coro}

\noindent {\bf Example: } We provide an 
example to illustrate this. Suppose that $\tilde{C}: 
\R_+ \times \cI \to \R$, where $\tilde{C}(r, a) = r 
\ L(a) + a = r \ L \ p(a) + a$ and $p(a) = \exp(- \xi \ a)$
for some $\xi > 0$. Clearly, $\tilde{C}$ is a mapping that
tells us the cost of an agent  seeing $r$ attacks as a 
function of its security investment.  Fix $r \in \R_+$ and 
differentiate $\tilde{C}(r, a)$ with respect to $a$. 
\beqan
\frac{\partial \tilde{C}(r, a)}{\partial a}
\myeq r \ L \ p'(a) + 1 \lb
\myeq - r \ L \ \xi \ \exp(- \xi \ a) + 1
\eeqan
This yields $I^{\opt}(r) = \min(I_{\max}, \max(I_{\min}, 
\log(r \ L \ \xi) / \xi))$, $r \in \R_+$. It is obvious 
that $I^{\opt}$ is nondecreasing in 
$r$. Also, $r_{\min} = \exp(I_{\min} \ \xi) 
/ (L \ \xi)$ and $r_{\max} = \exp(I_{\max} \ \xi) 
/ (L \ \xi)$. Substituting these expressions in the
given functions, we obtain
\beqan
p^\star(r) 
\myeq p(I^{\opt}(r))
	= \big( \tilde{R}(r) \ L \ \xi \big)^{-1},  
\eeqan
where $\tilde{R}(r) = \min(r_{\max}, \max(r_{\min}, r))$. 
Thus, the infection probability at the optimal investment
is decreasing in the expected number of attacks $r$. 
\myskip

\noindent $\bullet$
{\bf The existence of pure-strategy Nash equilibrium: }
Let $\Delta_n$,
$n \in \N$, denote the probability simplex in $\R^n$. 
The following lemma establishes the existence of 
{\em pure-strategy} NEs of population games.
In order to improve readability, we defer the proofs
of all main results to Section~\ref{sec:Proofs}, which
can be skipped without causing confusion elsewhere.
\myskip

\begin{lemma}	\label{lemma:existence}
For every pair of population size vector $\bs \in 
\Delta_{D_{\max}}$ and mixing vector $\bg \in 
\R_+^{D_{\max}}$, there exists a pure-strategy NE 
of the corresponding population game.
\end{lemma}
\begin{proof}
A proof is provided in Section~\ref{appen:existence}. 
\end{proof}

From an earlier discussion, under Assumption~\ref{assm:pa}, any 
NE of a population game, say $\bx^\star$, is a pure-strategy
NE. In other words, there exists 
a pure strategy profile $\ba^\star$ such that 
$\bx^\star(\{\ba^\star\}) = 1$. This is because, once the REs
$e_d(\bx^\star, \bs), d \in \cD,$ are fixed at the NE, every 
population has a unique optimal investment that minimizes 
its cost given by (\ref{eq:Cost}).

\section{Main Analytical Results}
	\label{sec:MainResults}

In the previous section, we established the existence
of a pure-strategy NE. But, when there are more
than one NE, it is not always obvious which NE is more
likely to emerge in practice, and one often has to turn 
to equilibrium selection theory in order to identify more 
likely NEs. If this were the case for our problem, 
it would be difficult to compare how the overall 
security would be affected by the varying degree 
correlations of the underlying dependence graph. 

Our first result addresses this issue and establishes 
the {\em uniqueness} of
pure-strategy NE of a population game. Thus, it allows 
us to compare the network security
at NEs as system parameters change. 
\myskip

\begin{theorem} \label{thm:unique}
Given a population size vector $\bs \in \Delta_{D_{\max}}$ 
and a mixing vector ${\bf g} \in \R_+^{D_{\max}}$, 
there exists a 
unique pure-strategy NE of the population game. 
\end{theorem}
\begin{proof}
A proof is provided in Section~\ref{appen:unique}. 
\end{proof}

We denote the unique pure-strategy NE in Theorem
\ref{thm:unique} by $\ba^\star(\bs, \bg)$ hereinafter. 
Our main result on the influence of degree 
correlations on network security is stated 
in the following theorem: it tells us how the effects
of network mixing 
captured via the mixing vector ${\bf g}$ might
change the security investments of strategic agents
and ensuing network security 
{\em in the local neighborhood of a neutral 
dependence graph}.

To make progress, we assume that $p^\star$ satisfies
the following assumption. 
\myskip 

\begin{assm} \label{assm:pstar}
The product $r \cdot \dot{p}^\star(r)$ is strictly 
increasing over $[r_{\min}, r_{\max}]$. 
\myskip
\end{assm}

\noindent For example, this assumption is true when the
optimal infection probability $p^\star$ can be well 
approximated over the interval [$r_{\min}, r_{\max}]$ by 
(a) $p^\star(r) = \nu_1 / 
(r + \nu_2)^{\chi_1}$ with $\nu_1, \chi_1 > 0$ and 
$0 \leq \nu_2 \leq r_{\min} / \chi_1$ 
or (b) $p^\star(r) = \nu_3 / 
\left( \log(r + \nu_4) \right)^{\chi_2}$ with $\nu_3, 
\chi_2 > 0$ and $\nu_4 \geq 1$ satisfying
\beqan
\frac{r_{\min}}{\log(r_{\min} + \nu_4)}
\geq \frac{\nu_4}{\chi_2 + 1}.
\eeqan 
Obviously, it holds when the optimal infection
probability can be approximated by a sum of these 
functions or other functions that satisfy the 
assumption. 

Let ${\bf 1}$ be the $D_{\max} \times 1$ vector
consisting of ones, i.e., ${\bf 1} = (1, \ldots ,
1)^T$. 
\myskip

\begin{theorem}	\label{thm:1}
Fix a population size vector $\bs \in \Delta_{D_{\max}}$
and assume that $\ba^\star(\bs, {\bf 1}) \in 
{\rm int}(\cI^{D_{\max}})$. 
Then, there exists an open, convex set ${\cal G}
\subset \R_+^{D_{\max}}$ containing ${\bf 1}$
such that if $\bg^i \in {\cal G}$, $i =$ 1, 2, are 
two mixing vectors satisfying
\beqa
\sum_{d'=1}^d w_{d'}(\bs) \ g^1_{d'}
\myleq \sum_{d'=1}^d w_{d'}(\bs)  \ g^2_{d'}
\ \mbox{ for all } \ d \in \cD,
	\label{eq:thm1}
\eeqa 
then $e_{\avg}(\ba^\star(\bs, \bg^2), \bs) 
\leq e_{\avg}(\ba^\star(\bs, \bg^1), \bs)$. 
Furthermore, if the inequality in (\ref{eq:thm1})
is strict for some $d \in \cD$, then 
$e_{\avg}(\ba^\star(\bs, \bg^2), \bs) 
< e_{\avg}(\ba^\star(\bs, \bg^1), \bs)$.
\end{theorem}
\begin{proof}
Please see
Section~\ref{appen:thm1} for a proof. 
\end{proof}

\change{A key idea in the proof of the theorem is 
the following: we construct a finite sequence 
of mixing vectors, starting with $\bg^2$ and ending 
with $\bg^1$. In each step, the RE experienced by 
agents in some population $d_1$ climbs while that 
of agents in another population $d_2 < d_1$ is
reduced proportionately. We show that this `transfer' 
of some of RE from agents with a smaller degree
($d_2$) to other agents with a larger degree ($d_1$) 
results in an increase in ARE at the unique 
pure-strategy NE. Moreover, we provide
a procedure for constructing such a sequence of
mixing vectors.}

From (\ref{eq:ARE}) and the definition of
mixing vector, an admissible mixing 
vector ${\bf g}$ must satisfy the following equality. 
\beqan
e_{\avg}(\bx, \bs)
\myeq \sum_{d \in \cD} \big( w_d(\bs) \cdot 
	e_d(\bx, \bs) \big) \lb
\myeq \sum_{d \in \cD} \big( w_d(\bs) \cdot g_d \cdot 
		e_{\avg}(\bx, \bs) \big) 
\eeqan
or, equivalently, $\sum_{d \in \cD} w_d(\bs) \cdot g_d = 1$. 
This implies that we can view ${\bf v}(\bs, \bg) = 
(v_d(\bs, \bg); \ d \in \cD)$ with $v_d(\bs, \bg) = w_d(\bs) 
\cdot g_d$ as a distribution over $\cD$.
When the inequality in (\ref{eq:thm1}) 
is strict for some $d \in \cD$ (i.e., $\bg^1 
\neq \bg^2$), it means that the distribution 
${\bf v}^1(\bs, \bg)$ first-order stochastically 
dominates ${\bf v}^2$(\bs, \bg)~\cite{SO}.\footnote{This
is equivalent to saying that a random variable with 
distribution ${\bf v}^1(\bs, \bg)$ is larger than a random
variable with distribution ${\bf v}^2(\bs, \bg)$ with 
respect to the usual stochastic order \cite{SO}.}
Hence, Theorem \ref{thm:1} states that the ARE increases 
as the distribution ${\bf v}(\bs, \bg)$ becomes 
(stochastically) larger. 

The following lemma provides a sufficient
condition for (\ref{eq:thm1}). 
\myskip

\begin{lemma}	\label{lemma:suff}
Suppose that two mixing vectors $\bg^1$ and 
$\bg^2$ satisfy
\beqa
\frac{g^1_d}{g^2_d} \leq \frac{g^1_{d+1}}{g^2_{d+1}}
\ \mbox{ for all } \ d = 1, 2, \ldots, D_{\max}-1.
	\label{eq:suff}
\eeqa
Then, the condition (\ref{eq:thm1}) in 
Theorem~\ref{thm:1} holds. 
\end{lemma}
\begin{proof}
Please see Section~\ref{appen:suff} for a proof. 
\end{proof}

An interpretation of (\ref{eq:suff}) is
that agents experience comparatively greater REs
with increasing degrees under mixing 
vector $\bg^1$ compared to under mixing vector
$\bg^2$. Thus, 
Theorem \ref{thm:1} tells us that, when agents 
face higher risks from their 
neighbors with increasing degrees, 
the resulting ARE at the pure-strategy
NE climbs.

\subsection{Case study - role of cost effectiveness of
	security measures}

As mentioned earlier, Theorem~\ref{thm:1} sheds some
light on how the changing degree correlations of
the underlying dependence graph might influence the
ARE as it deviates from a neutral 
graph and becomes either assortative or disassortative.
Interestingly, it turns out that the answer also 
depends on the (cost) effectiveness of available security 
measures, i.e., how quickly the infection probability
$p$ drops with security investment. 
To illustrate this, we consider following example
cases. 

Suppose that $p^\star(r) = \nu \ r^{-\chi}$ 
over $[r_{\min}, r_{\max}]$ for some 
$\nu, \chi > 0$.  
\myskip

{\bf Case 1: Effective security measures -- 
$\chi > 1$:}  This describes cases
where the security measures are cost effective
in that the probability of infection falls
quickly with increasing security investments. 
In this case, it is easy to see that the expected 
number of {\em successful} attacks or infections 
an agent of degree $d$
suffers at an NE, namely $\big( \tau_A + d 
\ e_{d}(\ba^\star(\bs, \bg), \bs) \big) 
p^\star\big( \tau_A + 
d \ e_{d}(\ba^\star(\bs, \bg), \bs) \big)$, is
decreasing in $d$ when the mixing vector
$\bg$ is sufficiently close to 
${\bf 1}$. Thus, at a pure-strategy NE, agents with 
higher degrees suffer {\em fewer} number of
infections than agents with smaller degrees.

For this reason, if the network is assortative, 
agents with higher degrees would see lower risks from 
their neighbors that tend to have larger degrees
as well. Accordingly, $g_d$ would decrease
with $d$, and Theorem~\ref{thm:1} suggests that 
the ARE would decrease (compared to the 
case with a neutral dependence graph). 
A similar argument tells us that if the network 
becomes disassortative and agents with higher degrees
tend to be neighbors with those of smaller degrees, 
the ARE would rise as a result. 
\myskip

{\bf Case 2: Ineffective security measures -- 
$\chi < 1$:} In contrast to the first
case, in the second case the probability of 
infection does not diminish rapidly with the security
investments. Consequently, agents with higher
degrees would suffer more infections
in spite of higher 
security investments because they also experience 
more attacks. Thus, Theorem~\ref{thm:1} indicates
that when the network is
assortative (resp. disassortative), the ARE would 
be higher (resp. lower) compared to the case of 
a neutral dependence graph. 
\myskip

This finding highlights another layer of difficulty 
in understanding the effects of network mixing on 
overall network security when the agents are 
strategic; the overall effects of degree correlations
depend also on how effective security measures
are at fending off attacks. Our finding suggests
that when the security measures are more cost
effective and the probability of infection 
drops quickly with increasing security
investments (case 1), the higher assortativity
of dependence graph tends to reduce the ARE
at the equilibrium. On the other hand, when the
security measures are not cost effective (case 2), 
it has the {\em opposite} effect.

Finally, we point out that our finding is proved 
only in the local neighborhood around the neutral 
dependence graph. However, as our numerical study in 
the subsequent section shows, we suspect that it
holds much more generally even outside the local 
neighborhood.

\change{
\section{Proofs of Main Results}
	\label{sec:Proofs}

This section contains the proofs of main 
results in Sections~\ref{sec:Preliminaries} and
\ref{sec:MainResults}. A reader who is not 
interested in the proofs can proceed to 
Section~\ref{sec:Numerical} for numerical
studies.  

\subsection{A proof of Lemma~\ref{lemma:existence}}
	\label{appen:existence}
	
Let $H: \cI^{D_{\max}} \to \cI^{D_{\max}}$, 
where $H_d(\ba) = I^{\opt}(\tau_A + d \ e(\ba, \bs))$, 
$d \in \cD$. Then, from Assumption~\ref{assm:pa} and the 
definition of $I^{\opt}$, the mapping $H$ is continuous.
Therefore, since $\cI^{D_{\max}}$ is a compact, convex subset 
of $\R^{D_{\max}}$, the Brouwer's fixed point theorem 
\cite{Kakutani} tells us that there exists a fixed point of 
$H$, say $\ba'$, 
such that $H(\ba') = \ba'$. It is clear from the definition 
of a pure-strategy NE in Definition~\ref{defn:PSNE}
that $\ba'$ is a pure-strategy NE.

\subsection{A proof of Theorem~\ref{thm:unique}}
	\label{appen:unique}

In order to prove the theorem, we will first prove that
if $\ba^1$ and $\ba^2$ are two pure-strategy NEs, then 
$e_{\avg}(\ba^1, \bs) = e_{\avg}(\ba^2, \bs)$. 	
We prove this by contradiction. Suppose that the 
claim is false and there exist two pure-strategy NEs
with different AREs. Without loss of generality, 
assume $e_{\avg}(\ba^1, \bs) < e_{\avg}(\ba^2, \bs)$. This 
means that $e_d(\ba^1, \bs) < e_d(\ba^2, \bs)$ for all $d \in 
\cD$. 

Together with Corollary~\ref{coro:OptInfProb}, 
this means $p(a^1_d) \geq p(a^2_d)$ for all $d \in \cD$ 
and, as a result,  
\beqan
e_\avg(\ba^1, \bs)
\myeq \Theta\left( \sum_{d \in cD} d \ s_d \ 
	p(a^1_d) \right) \lb
\mygeq \Theta\left( \sum_{d \in cD} d \ s_d \ 
	p(a^2_d) \right)
	= e_\avg(\ba^2, \bs). 
\eeqan
But, this contradicts the earlier assumption 
$e_{\avg}(\ba^1, \bs) < e_{\avg}(\ba^2, \bs)$. 
The theorem now follows from the observation that,
for every population $d \in \cD$,  given a fixed RE, 
there exists a unique optimal investment that 
minimizes the cost. This proves the uniqueness of
pure-strategy NE.

\subsection{A proof of Theorem~\ref{thm:1}}
	\label{appen:thm1}
	
Since the population size vector $\bs$ is fixed, for
notational convenience, we shall omit the dependence
of $e_{\avg}$, $\Phi$ and $\bw$ on $\bs$ throughout 
the proof.
 
First, note from (\ref{eq:ARE}) that pure-strategy NEs 
$\ba^i = \ba^\star(\bs, \bg^i)$, $i = 1, 2$, satisfy
\beqa
e_{\avg}(\ba^i)
\myeq \Phi\left( \sum_{d \in \cD} w_d \ p(a^i_d) 
	\right) \lb
\myeq \Phi\left( \sum_{d \in \cD} w_d \ 
	p^\star \big( \tau_A + d 
	\ g_d^i \ e_{\avg}(\ba^i) \big) \right).
	\label{eq:2-1}
\eeqa
Moreover, given a mixing vector $\bg$, by the uniqueness 
of pure-strategy NE and Corollary~\ref{coro:OptInfProb}, 
there exists a unique $e_{\avg}$
that satisfies (\ref{eq:2-1}), namely 
$e_{\avg}(\ba^\star(\bs, \bg))$. 

Define $\vartheta: \R_+^{D_{\max}} \times \R_+ 
\to \R$, where 
\beqa
\vartheta(\bg, e) 
\myeq \Phi\left( \sum_{d \in \cD} w_d 
	\ p^\star(\tau_A + d \ g_d \ e) \right) - e. 
	\label{eq:vartheta}
\eeqa
From (\ref{eq:2-1}), we have 
\beqa
\vartheta(\bg^i, e_{\avg}(\ba^i)) = 0, \ i = 1, 2.
	\label{eq:2-2}
\eeqa
Also, one can verify 
\beqa
\frac{\partial \vartheta(\bg^i, e_{\avg}(\ba^i))}
	{\partial e} < 0.
	\label{eq:2-2a}
\eeqa
This is intuitive because as the ARE rises, 
agents see higher risks and invest more in security, 
thus reducing their vulnerability to attacks. 

From (\ref{eq:2-2}) and (\ref{eq:2-2a}) and the assumption
in the theorem, the implicit 
function theorem~\cite{Rudin} tells us that there exist 
open sets $O_{e} \in \R_+$ and $O_{\bg} \subset 
\R_+^{D_{\max}}$, which contains ${\bf 1}$, 
and a function $e^\star:O_{\bg} \to O_{e}$ 
such that, for all $\bg \in O_{\bg}$, 
\beqan
\vartheta(\bg, e^\star(\bg)) = 0. 
\eeqan
It is clear that $e^\star(\bg) = e_{\avg}(\ba^\star(\bs, 
\bg))$ for all $\bg \in O_{\bg}$. 
In addition, for all $d \in \cD$, 
\beqa
\frac{\partial e^\star(\bg)}{\partial g_d}
\myeq - \left( \frac{\partial \vartheta(\bg, e^\star(\bg)) }
	{\partial e} \right)^{-1} 
	\frac{\partial \vartheta(\bg, e^\star(\bg))}{\partial g_d}. 
	\label{eq:2-3}
\eeqa
Hence, (\ref{eq:2-3}) tells us how the ARE will change 
locally as the mixing vector is perturbed around ${\bf 1}$, 
i.e., a neutral graph. 

The theorem can be proved with the help of 
the following lemma. Let ${\bf e}_d$ denote the 
$D_{\max} \times 1$ zero-one vector whose only nonzero
element is the $d$th entry. 
\myskip

\begin{lemma}	\label{lemma:2}
Let $1 \leq d_2 < d_1 \leq D_{\max}$. 
Choose $\bg^3 \in O_{\bg}$ and $\delta >0$. 
Suppose  
$\bg^4 := \bg^3 + \delta {\bf e}_{d_1} - \delta 
\frac{w_{d_1}}{w_{d_2}} {\bf e}_{d_2} \in O_{\bg}$. 
Then, for all sufficiently small $\delta$, we have
$e_{\avg}(\ba^\star(\bs, \bg^3)) < e_{\avg}(
\ba^\star(\bs, \bg^4))$. 
\end{lemma}
\begin{proof}
A proof of lemma is provided in 
Section~\ref{appen:lemma2}.
\end{proof}

Theorem~\ref{thm:1} now follows from the observation that, 
starting with mixing vector $\bg^2$, we can obtain the 
other mixing vector $\bg^1$ by performing a sequence
of operations described in Lemma~\ref{lemma:2}. We 
first provide the procedure for general cases and then 
illustrate it using an example.
\myskip

\noindent
{\underline{\bf Procedure for constructing ${\bf g}^1$ 
from ${\bf g}^2$}} \myskip

\bitem
\item {\bf Step 0:} Let $\tilde{\bg} = \bg^2$. 

\item {\bf Step 1:} Find 
$d_1 = \max\{d \in \cD \ | \ g^1_{d} > \tilde{g}_d\}$
and 
$d_2 = \min\{d \in \cD \ | \ \tilde{g}_{d} > g^1_d\}$.

\item {\bf Step 2:} Increase $\tilde{g}_{d_1}$ by
$\min \big( g^1_{d_1} - \tilde{g}_{d_1}, 
\frac{w_{d_2}}{w_{d_1}} (\tilde{g}_{d_2} - g^1_{d_2}) 
\big)$,
and reduce $\tilde{g}_{d_2}$ by
$\min \big( \tilde{g}_{d_2} - g^1_{d_2}, 
\frac{w_{d_1}}{w_{d_2}} (g^1_{d_1} - \tilde{g}_{d_1}) \big)$.

\item {\bf Step 3:}
If $\tilde{\bg} \neq \bg^1$, repeat Steps 1 and 2
with new $\tilde{\bg}$.
\myskip

\eitem

The first-order stochastic dominance of ${\bf v}^1 := 
(w_d \ g^1_d; \ d \in \cD)$ over ${\bf v}^2 := 
(w_d \ g^2_d; \ d \in \cD)$ guarantees that the above
procedure will terminate after a finite number of iterations
with $\tilde{\bg} = \bg^1$. Moreover, Lemma~\ref{lemma:2}
tells us that the ARE increases after each iteration. 
\myskip

{\bf Example --}
Suppose ${\bf w} = (0.6 \ 0.3 \ 0.1)^T$, $\bg^1 = 
(0.942 \ 1.05 \ 1.2)^T$, and $\bg^2 = (1.02 \ 1.0
\ 0.88)^T$. Then, one can easily verify that
condition (\ref{eq:thm1}) in Theorem
\ref{thm:1} is satisfied. 
\myskip

\noindent $\circ$ Step 0: $\tilde{\bg} = \bg^2$.  
\myskip

\noindent \underline{{\em Iteration \#1}}

\noindent $\circ$ Step 1: $d_1 = 3$ and $d_2 = 1$. 

\noindent $\circ$ Step 2:
Increase $\tilde{g}_3$ by $\min\big( 1.2 - 0.88, 
\frac{0.6}{0.1} \times (1.02 - 0.942) \big) = 0.32$, 
and decrease $\tilde{g}_1$ by $\min\big( 1.02 - 0.942, 
\frac{0.1}{0.6} \times (1.2 - 0.88) \big) = 0.053$.
This gives us new $\tilde{\bg} = (0.967 \ 1.0 \ 
1.2)^T$, which does not equal $\bg^1$.
\myskip

\noindent \underline{{\em Iteration \#2}}

\noindent $\circ$ Step 1: $d_1 = 2$ and $d_2 = 1$. 

\noindent $\circ$ Step 2:
Increase $\tilde{g}_2$ by $\min\big( 1.05 - 1.0, 
\frac{0.6}{0.3} \times (0.967 - 0.942) \big) = 
0.05$, and reduce $\tilde{g}_1$ by 
$\min\big( 0.967 - 0.942, \frac{0.3}{0.6} \times 
(1.05 - 1.0) \big) = 0.025$. This 
yields new $\tilde{\bg} = (0.942 \ 1.05 \ 1.2)^T$, 
which is equal to $\bg^1$, and we terminate the
procedure.

\subsection{A proof of Lemma~\ref{lemma:suff}}
	\label{appen:suff}
	
We prove the lemma with help of the following 
Lemma~\ref{lemma:suff}, whose proof is straightforward
and is omitted here.
\myskip

\begin{lemma} \label{lemma:aux}
Suppose that ${\bf a} = (a_\ell; \ell = 1, \ldots, K)$
and ${\bf b} = (b_\ell;\ell = 1, \ldots, K)$ are two 
finite sequences of nonnegative real numbers of length
$K > 1$ and satisfy 
\beqa
\frac{b_{\ell+1}}{a_{\ell+1}} \leq \frac{b_\ell}{a_\ell}
	\ \mbox{ for all } \ell = 1, \ldots, K-1. 
	\label{eq:lemma2-0}	
\eeqa
Then, 
\beqa
\frac{ \sum_{\ell = 1}^{K} b_\ell }
	{ \sum_{\ell = 1}^{K} a_\ell } 
\myleq
\frac{ \sum_{\ell = 1}^k b_\ell }{ \sum_{\ell = 1}^k a_\ell }
	\ \mbox{ for all } k = 1, \ldots, K.
	\label{eq:lemma2-1}
\eeqa \\ \vspace{-0.1in}
\end{lemma}

Proceeding with the proof of Lemma~\ref{lemma:suff}, 
recall that the condition (\ref{eq:suff}) in 
Lemma~\ref{lemma:suff} states
\beqa
&& \myhb \frac{ w_{d+1}(\bs) \ g^2_{d+1} }
	{ w_{d+1}(\bs) \ g^1_{d+1} } 
	= \frac{v_{d+1}(\bs, \bg^2)}{v_{d+1}(\bs, \bg^1)}
\leq \frac{v_{d}(\bs, \bg^2)}{v_{d}(\bs, \bg^1)}
	=\frac{ w_{d}(\bs) \ g^2_{d} }
	{ w_{d}(\bs) \ g^1_{d} }  \lb
&& \hspace{0.5in} \mbox{ for all } d = 1, 2, \ldots, 
	D_{\max}-1.
	\label{eq:lemma2-4}
\eeqa
Together with (\ref{eq:lemma2-4}), Lemma~\ref{lemma:aux} 
tells us, for all $d = 1, 2, \ldots, D_{\max}$, we have
\beqan
\frac{\sum_{d' = 1}^{D_{\max}} v_{d'}(\bs, \bg^2)}
	{\sum_{d' = 1}^{D_{\max}} v_{d'}(\bs, \bg^1)} 
= \frac{1}{1}
\leq \frac{\sum_{d' = 1}^{d} v_{d'}(\bs, \bg^2)}
	{\sum_{d' = 1}^{d} v_{d'}(\bs, \bg^1)} 
\eeqan
or, equivalently, 
\beqan
\sum_{d' = 1}^d v_{d'}(\bs, \bg^1)
\leq  \sum_{d' = 1}^d v_{d'}(\bs, \bg^2). 
\eeqan
This completes the proof of the lemma.

\subsection{A proof of Lemma~\ref{lemma:2}}
	\label{appen:lemma2}
	
In order to prove the lemma, we will use 
(\ref{eq:2-3}) to demonstrate 
\beqa
0 > \frac{\partial e^\star(\bg^3)}{\partial g_{d_1}}
>  \frac{w_{d_1}}{w_{d_2}}  
	\frac{\partial e^\star(\bg^3)}{\partial g_{d_2}}.
	\label{eq:2-4}
\eeqa 
First, note 
\beqan
\frac{\partial \vartheta(\bg^3, e)}{\partial e}
\myeq \dot{\Phi}\left( \sum_{d' \in \cD} w_{d'} \ 
	p^\star(\tau_A + d' \ g^3_{d'} \ e) \right) \lb
&& \myb \times \left( \sum_{d' \in \cD} w_{d'} \ 
	\dot{p}^\star(\tau_A + 
	d' \ g^3_{d'} \ e) \ d' \ g^3_{d'} \right) - 1 \lb
& \myb < & \myb 0.  
\eeqan
Hence, in order to prove (\ref{eq:2-4}), it suffices to 
show 
\beqan
0 > \frac{\partial \vartheta(\bg^3, e^\star(\bg^3))}{\partial g_{d_1}}
> \frac{w_{d_1}}{w_{d_2}} 
	\frac{\partial \vartheta(\bg^3, e^\star(\bg^3))}{\partial g_{d_2}}. 
\eeqan

From the definition of $\vartheta$ in (\ref{eq:vartheta}), 
\beqa
\frac{\partial \vartheta(\bg^3, e)}{\partial g_{d}}
\myeq \dot{\Phi}\left( \sum_{d' \in \cD} w_{d'} \ 
	p^\star(\tau_A + d' \ g^3_{d'} \ e) \right) \lb
&& \times \ w_d \ \dot{p}^\star(\tau_A 
	+ d \ g^3_d \ e) \ d \ e.
	\label{eq:2-5}
\eeqa
Thus,
\beqa
&& \myhb  \frac{\partial \vartheta(\bg^3, e^\star(\bg^3))}
	{\partial g_{d_1}}
- \frac{w_{d_1}}{w_{d_2}} 
	\frac{\partial \vartheta(\bg^3, e^\star(\bg^3))}
		{\partial g_{d_2}} \lb
\myeq \dot{\Phi}\left( \sum_{d' \in \cD} w_{d'} \ 
	p^\star(\tau_A + d' \ g^3_{d'} \ e^\star(\bg^3)) \right) 
		\ w_{d_1} \ e^\star(\bg^3) \lb
&& \times \Big( d_1 \ \dot{p}^\star(\tau_A + d_1 \ g^3_{d_1} 
	\ e^\star(\bg^3)) \lb
&& \myhf -\ d_2 \ \dot{p}^\star(\tau_A + d_2 \ g^3_{d_2} 
		\ e^\star(\bg^3)) \Big) \lb
\myg 0, 
	\label{eq:2-6}
\eeqa
where the inequality follows from our assumption $d_2 < 
d_1$, $g^3_{d_1} \approx g^3_{d_2}$ for sufficiently small 
set $O_{\bg}$ including ${\bf 1}$ and Assumption
\ref{assm:pstar}. 

Putting things together, 
\beqa
&& \myhb e_{\avg}(\ba^\star(\bs, \bg^4), \bs) 
	- e_{\avg}(\ba^\star(\bs, \bg^3), \bs) \lb
\myeq - \left( \frac{\partial \vartheta(\bg^3, e^\star(\bg^3))}
	{\partial e} \right)^{-1}  \lb
&& \myf \times 
	\left( \frac{\partial \vartheta(\bg^3, e^\star(\bg^3))}
			{\partial g_{d_1}}
	\delta - \frac{w_{d_1}}{w_{d_2}} 
	\frac{\partial \vartheta(\bg^3, e^\star(\bg^3))}
		{\partial g_{d_2}} \delta \right) \lb
&& + \ o(\delta). 
	\label{eq:2-7}
\eeqa
From the inequality in (\ref{eq:2-6}), for all sufficiently small 
$\delta > 0$, we have (\ref{eq:2-7}) $> 0$. 
}


\section{Numerical Results}
	\label{sec:Numerical}
	
In this section, we provide some numerical results (i) 
to validate our main findings in the previous section 
and (ii) to illustrate how the cost effectiveness of 
available security measures and the function $\Theta$ 
in (\ref{eq:ARE1}) affect the resulting ARE at the 
pure-strategy NE. 
While our analytical findings in the previous
section offer some insights into the {\em qualitative}
behavior of the network security measured by ARE, 
it does not provide {\em quantitative} answers. For this
reason, we resort to numerical studies to find out
how the effectiveness of security measures and 
the sensitivity of ARE to agents' vulnerability to
attacks shape the impact of degree correlations on 
network security.

For the numerical results, the maximum degree is set to 
$D_{\max} = 20$, and the population size vector is 
assumed to be a (truncated) power law with exponent
2, i.e. $f_d \propto d^{-2}$. It is shown that the degree
distribution of many natural and engineered networks can 
be approximated using a power law with exponents in 
[1, 3] (e.g., \cite{Albert2000, 
Lakhina2003}). In addition, we choose $\tau_A = 0.7$, 
$\beta_{IA} = 1$, $I_{\min} = 10^{-3}$, $I_{\max} = 10^3$, 
and $L = 10$. Here, we intentionally pick small $I_{\min}$
and large $I_{\max}$ so that neither becomes an active 
constraint at an NE. 

The mixing vectors we consider are of the form $g_d^{(\rho)} 
\propto d^{\rho}, d \in \cD$, with $\rho \in$ [-0.3, 0.3], 
subject to the constraint $\sum_{d \in \cD} w_d \cdot 
g^{(\rho)}_d = 1$. We pick this range of
$\rho$ to clearly demonstrate the 
behavior of ARE as a function of $\rho$. 
Obviously, when $\rho = 0$, we have 
$g_d^{(0)} = 1$ for all $d \in \cD$ and the dependence graph
is neutral. Note that if $\rho_2 < \rho_1$, we have
\beqan
\frac{g^{(\rho_2)}_{d+1}}{g^{(\rho_2)}_d}
	= \left( \frac{d+1}{d} \right)^{\rho_2} 
\myl \left( \frac{d+1}{d} \right)^{\rho_1}
	= \frac{g^{(\rho_1)}_{d+1}}{g^{(\rho_1)}_d} \lb
&& \mbox{ for all } d = 1, 2, \ldots, D_{\max}-1,
\eeqan
and the sufficient condition in (\ref{eq:suff}) holds 
with strict inequality for $\bg^i = \bg^{(\rho_i)}$, 
$i = 1, 2$. 
Consequently, as $\rho$ ascends, agents experience 
greater REs with increasing degrees. Finally, 
the interval [-0.3, 0.3] provides a sufficiently
wide range of mixing vectors to illustrate that 
the qualitative nature of our analytical findings
in the previous section 
holds outside a small local neighborhood around the
neutral graph.

We assume infection probability $p(a) = \epsilon^\gamma / 
(a + \epsilon)^\gamma$, where $\epsilon = 0.1$. We 
vary $\gamma$ to alter the cost effectiveness of security 
measures; the larger $\gamma$ is, the more cost effective
they are in that the infection probability diminishes 
faster with security investments. After a little algebra, 
we get 
\beqa
I^{\opt}(r)
\myeq (r \ L \ \epsilon^\gamma)^{\frac{1}{\gamma+1}} 
	- \epsilon, \ r \in [r_{\min}, r_{\max}], 
	\label{eq:Iopt}
\eeqa
where 
\beqan
r_{\min} = \frac{(\epsilon + I_{\min})^{\gamma+1}}
	{L \ \epsilon^\gamma}
\ \mbox{ and } \
r_{\max} = \frac{(\epsilon + I_{\max})^{\gamma+1}}
	{L \ \epsilon^\gamma}. 
\eeqan
Substituting (\ref{eq:Iopt}) in $p(a)$ yields 
\beqan
p^\star(r)
\myeq \frac{\epsilon^\gamma}{(r \ L \ 
	\epsilon^\gamma)^{\gamma/(\gamma+1)} }, 
	\ r \in [r_{\min}, r_{\max}]. 
\eeqan
Therefore, $p^\star(r) \propto r^{- \gamma/(\gamma+1)}$
over the interval $[r_{\min}, r_{\max}]$, and the infection
probability at the optimal investment falls more quickly 
with an increasing risk as $\gamma$ climbs.

\subsection{Effects of infection probability function}

In our first numerical study, we examine how the 
effectiveness of security measures, which is determined 
by $\gamma$, shapes the effects of dependence graph 
assortativity on equilibrium ARE. 
Since $\frac{\gamma}{\gamma+1} < 1$, this corresponds to
case 2 discussed in the previous section. As a result,
when $\rho$ is negative (resp. positive), the dependence 
graph is disassortative (resp. assortative), and 
Theorem~\ref{thm:1} suggests that the ARE shall rise 
with increasing $\rho$. However, the theorem does 
not tell us the {\em quantitative} behavior of the 
equilibrium ARE as either $\rho$ or 
the parameter of infection probability, 
namely $\gamma$, changes. Thus, we turn to 
numerical studies to find an answer. 
For our study, we employ a linear ARE function 
in (\ref{eq:ARE2}) with $K \cdot d_{\avg}(\bs)
= K \cdot 2.254 = 1000$. 

\begin{figure}[h]
\centerline{
	\includegraphics[width=3.3in]{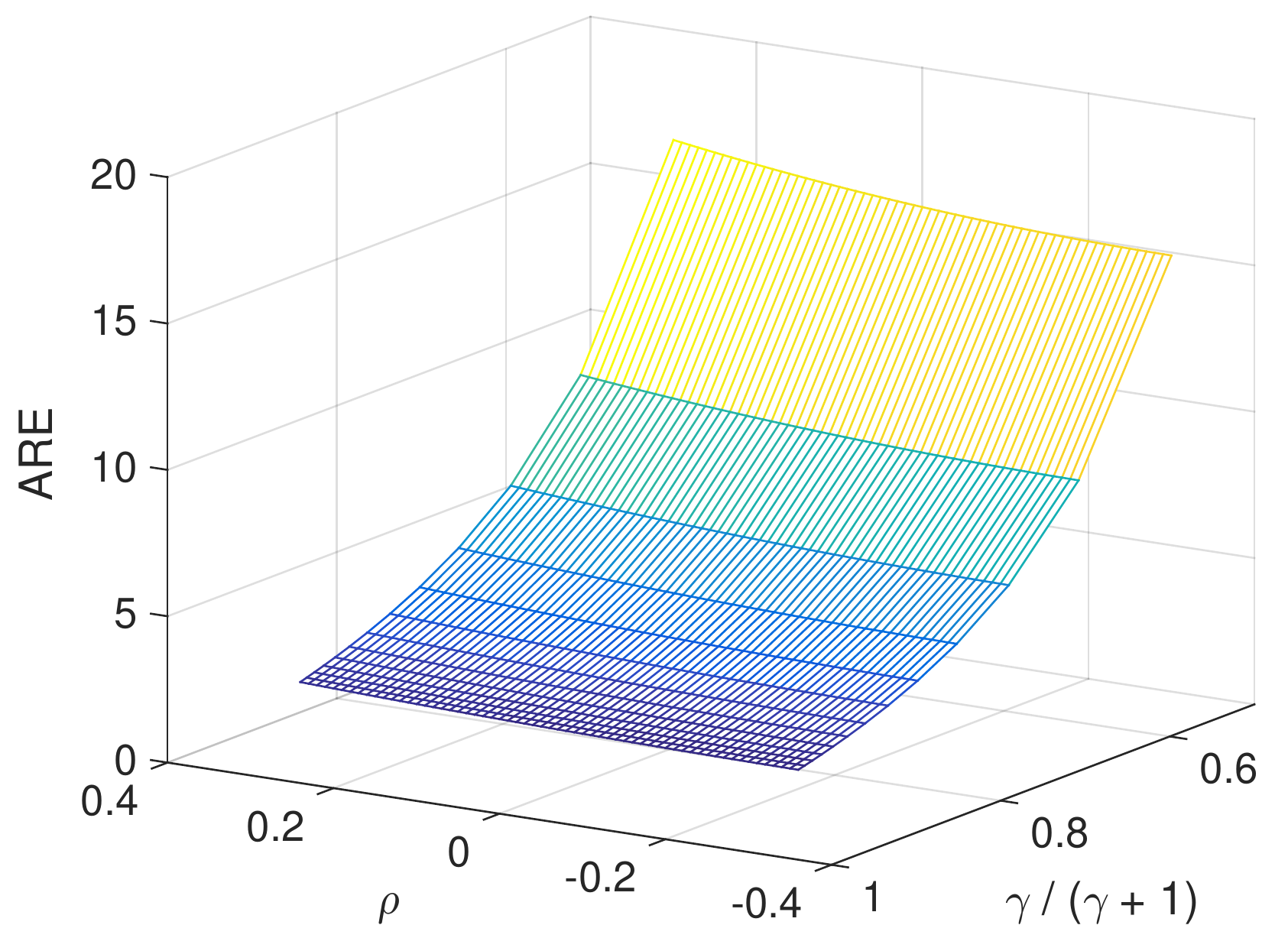}
}
\caption{Plots of ARE as a function of $\gamma$
	and $\rho$.}
\label{fig:mesh}
\end{figure}

\begin{figure}[h]
\centerline{
	\includegraphics[width=3.3in]{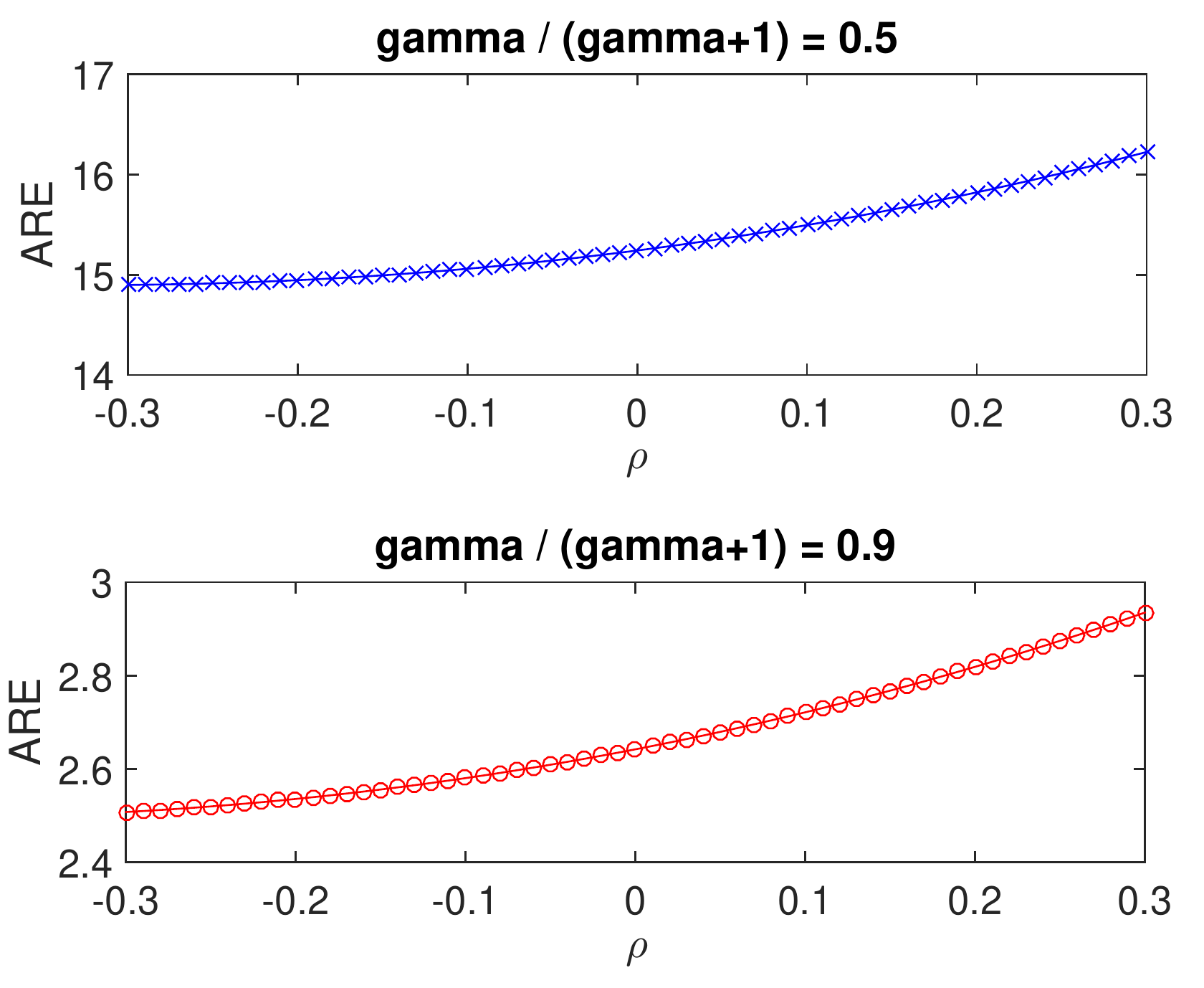}
}
\caption{Plots of ARE as a function of $\rho$ 
	for $\gamma = 1$ and $\gamma = 9$.}
\label{fig:two}
\end{figure}

Fig.~\ref{fig:mesh} plots the ARE at the pure-strategy NE
as both the parameters $\rho$ and $\gamma$ are varied. 
There are two observations that we point out. 
First, it confirms that, for fixed $\gamma$,  
the ARE rises with increasing $\rho$ as predicted by
Theorem~\ref{thm:1}. This can be seen more easily 
in Fig.~\ref{fig:two}, which displays the ARE as a 
function of $\rho$ for two different values of $\gamma$ 
($\gamma
= 1$ and 9). 
Second, it is clear from Fig.~\ref{fig:mesh} that as 
$\gamma$ increases, hence
$\gamma / (\gamma + 1)$ climbs, the ARE
decreases quickly for all values of $\rho$ we considered.
This hints at high sensitivity of the equilibrium ARE
with respect to the cost effectiveness of available
security measures.

In addition to corroborating Theorem 
\ref{thm:1}, Fig.~\ref{fig:two} reveals two 
additional interesting
observations. First, it illustrates that the influence
of network mixing (equivalently, parameter $\rho$) on 
ARE is more pronounced when the security measures 
are more cost effective (i.e., $\gamma$ is larger); when 
$\gamma = 1$ (resp. $\gamma = 9$), the ARE rises from
14.9 to 16.23 (resp. from 2.508 to 2.935) as $\rho$
ascends from -0.3 to 0.3, which is roughly an 8.9
percent increase (resp. a 17 percent increase). Therefore, 
they indicate that, although the equilibrium AREs are
smaller when the security measures are more cost
effective, they also become more sensitive to the bias 
in REs caused by assortativity.  

Second, the ARE is a {\em convex} function of 
$\rho$. This hints that the impact of degree 
correlations on ARE gets stronger as the dependence
graph becomes more assortative. As a result, a drop
in ARE a disassortative dependence graph enjoys
may not be as large as an increase in ARE an assortative
dependence graph suffers. This in turn suggests
that social networks, which in general exhibit 
non-negligible positive degree correlations
\cite{Newman2002, Newman2003}, may
experience significant deterioration in security 
relative to the findings obtained using 
neutral networks.

\subsection{Effects of ARE function $\Theta$}

In our second study, we explore how the ARE
function $\Theta$ in (\ref{eq:ARE1}) 
affects equilibrium ARE. In particular, we are
interested in how sensitive the ARE is to the
assortativity of dependence graph as we vary the
shape of the function $\Theta$. 

\begin{figure}[h]
\centerline{
	\includegraphics[width=3.3in]{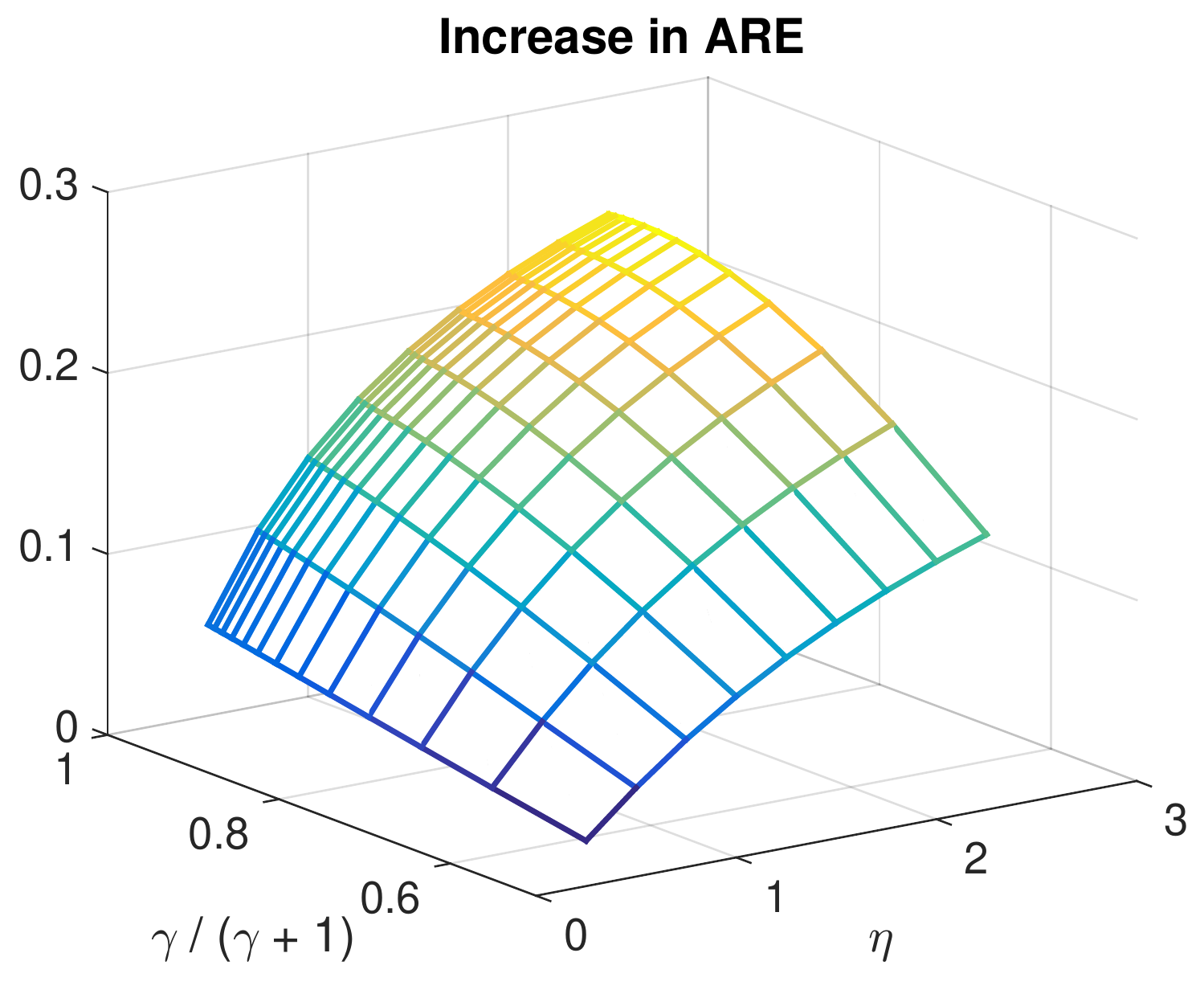}
}
\caption{Plots of the increase in ARE as a function of $\gamma$
	and $\eta$.}
\label{fig:increase-mesh}
\end{figure}

To this end, we adopt a family of functions of the form 
$\Theta(z) = K'(\eta) z^\eta$ with $K'(\eta), \eta > 0$, where 
the parameter $\eta$ is used to change the shape of the 
function $\Theta$. 
In order to compare the ARE as we vary $\eta$, we adjust 
the value of parameter $K'(\eta)$ as a function of $\eta$ so 
that the equilibrium ARE is identical under the neutral 
dependence graph (i.e., $\rho = 0$) with $\gamma = 5$ 
(equivalently, $\gamma / (\gamma + 1) = 0.8\bar{3}$)
for all values of $\eta$ we consider.

Fig.~\ref{fig:increase-mesh} plots the increase in ARE as we 
vary $\rho$ from -0.3 to 0.3 for different values of $(\gamma, 
\eta)$. More precisely, each point in the figure represents
the difference in ARE for $\rho =$ 0.3 and -0.3, divided by 
the value of ARE for $\rho = -0.3$. 

It is clear from Fig.~\ref{fig:increase-mesh} that 
when $\gamma / (\gamma + 1)$ is larger (indicating that 
the security measures are more cost effective), the ARE 
is more sensitive to assortativity because the relative 
increase in ARE is greater for all considered values of 
$\eta$. This confirms our finding
in the previous subsection (illustrated by Fig. 
\ref{fig:two}). 

More importantly, Fig.~\ref{fig:increase-mesh}  
reveals that assortativity has greater 
impact on ARE when the ARE function $\Phi$ is more sensitive 
to the vulnerability of neighbors summarized by $\sum_{d \in 
\cD} w_d(\bs) \ p(a_d)$. This observation is somewhat 
intuitive; as ARE becomes more sensitive to the vulnerability
of agents, any changes in the security investments of 
agents will likely amplify the effects other parameters,
including the assortativity of dependence graph.

\section{Conclusion}
	\label{sec:Conclusion}

We studied the effects of degree correlations on network security
in IDS. Our findings reveal that the network security 
degrades when agents with larger degrees experience higher
risks than those with smaller degrees. Moreover, somewhat 
unexpectedly, the cost effectiveness of available security 
measures determines how network mixing influences network 
security. Finally, our numerical studies suggest that as
the infection probability or the vulnerability of neighboring
agents becomes more sensitive to security investments, 
assortativity exerts greater impact on network security. 
Our analytical study carried out only a 
local analysis around neutral 
dependence graphs. We are currently working to generalize
our results beyond the local neighborhood of neutral 
graphs.

\end{document}